\documentclass[11pt,  onecolumn]{article}
\usepackage[height=22cm,width=17.5cm]{geometry}
\usepackage{pstricks}
\usepackage{mathrsfs}
\usepackage{enumerate}
\usepackage{bbm}
\usepackage{tikz}
\usepackage{pgfplotstable}
\usepackage{algorithm}
\usetikzlibrary{arrows}
\usepackage{relsize}
\usepackage{subcaption}
\usepackage{pgfplots}
\usepackage{scrextend}
\usepackage{lipsum}
\usepackage{epsfig}
\usepackage[cmex10]{amsmath}
\usepackage{amssymb}
\usepackage{amsmath}
\usepackage{amsthm}
\usepackage{amsfonts}
\usepackage{bm}
\usepackage{cite}
\usepackage{algpseudocode}
\usepackage{graphicx}
\usepackage{lipsum}
\usepackage{tikz}
\usepackage{mathrsfs} 
\usepackage{tikz}
\usepackage{url}
\usepackage{epsfig}
\usepackage[cmex10]{amsmath}
\usepackage{amssymb}
\usepackage{amsthm}
\usepackage{amsfonts}
\usepackage{bm}
\usepackage{cite}
\usepackage{algpseudocode}
\usepackage{graphicx}
\usepackage{lipsum}
\usepackage{tikz} 
 \usepackage{scrextend} 
 
\newtheorem{example}{Example}
\DeclareMathOperator*{\cut}{cut}

\newtheorem{theorem}{Theorem}
\newtheorem{corollary}{Corollary}
\newtheorem{lemma}{Lemma}
\newtheorem{remark}{Remark}
\newtheorem{definition}{Definition}

\begin{document}

\date{}
\title{Optimum Transmission Delay for Function Computation in NFV-based Networks: the role 
of Network Coding and Redundant Computing }
\author{Behrooz Tahmasebi~~~ Mohammad Ali Maddah-Ali~~~ Saeedeh Parsaeefard\\~~~Babak Hossein Khalaj
\thanks{Behrooz Tahmasebi, Mohammad Ali Maddah-Ali, and Babak Hossein Khalaj are  with the Department of Electrical Engineering, Sharif University of Technology, Tehran, Iran (emails:behrooz.tahmasebi@ee.sharif.edu, maddah\_ali@sharif.edu,  khalaj@sharif.edu). 
 Babak Hossein Khalaj is also with  the School of Computer Science, Institute for Research in Fundamental Sciences (IPM), Tehran, Iran.
Saeedeh Parsaeefard is with the Department of Communication Technologies,
Iran Telecommunication Research Center (ITRC), Tehran, Iran (email:
s.parsaeifard@itrc.ac.ir).}
}
 \renewcommand\footnotemark{}
\maketitle

\begin{abstract}

In this paper, we study the problem of delay minimization in NFV-based networks. In such systems, the ultimate goal of any request is to compute a sequence of functions in the network, where each function
can be computed at only a specific subset of network nodes. In conventional approaches, for each function, we choose one node from the corresponding subset of the nodes to compute that function. In contrast, in this work, we allow each function to be computed in more than one node, redundantly in parallel, to respond to a given request. We argue that such redundancy in computation not only improves the reliability of the network, but would also, perhaps surprisingly, reduce the overall transmission delay.
In particular, we establish that by judiciously choosing the subset of nodes which compute each function, in conjunction with a linear network coding scheme to deliver the result of each computation, we can characterize and achieve the optimal end-to-end transmission delay.
 In addition, we show that using such technique, it is possible to significantly reduce the transmission delay as compared to the conventional approaches. 
 In fact, in some scenarios, such reduction can even scale with the size of the network,
 where  by increasing the number of nodes that can compute the given function in parallel by a multiplicative factor, the end-to-end delay will also decrease by the same factor.
Moreover, we show that while finding the subset of nodes for each computation, in general, is a complex integer  program, approximation algorithms can be proposed to reduce the computational complexity. In fact, for the case where the number of computing nodes for a given function is upper-bounded by a constant, a dynamic programming scheme can be proposed to find the optimum subsets in polynomial times. Our numerical simulations confirm the achieved gain in performance in comparison with conventional approaches.

\end{abstract}

\textbf{Index terms$-$}
Delay-computation  trade-off, 
network coding, 
network function virtualization (NFV), 
network optimization, 
redundancy,
reliability.

\section{Introduction}

 Network function virtualization (NFV) is the state-of-the-art  architecture for future  data networks.
 NFV is an enabler to network slicing and cloud over core in 5G which can considerably improve the efficiency of resource utilization \cite{taleb1,taleb2,taleb3,taleb4,taleb5,taleb6,taleb7}.
  In NFV-based networks, network functions are virtualized and  the resulting vitrual network functions (VNFs) can be computed at any node  where they are placed. This modification in the architecture of data networks leads to  high flexibility and improvement in performance measures \cite{nfvra}.  
Consequently, resource management will become  a key problem in NFV realizations. In the resulting network function virtualization resource allocation (NFV-RA) problem, new challenging sub-problems such as  chain composition, forwarding graph embedding, and  scheduling will arise and need to be addressed efficiently\cite{nfvra}.

In such networks, the ordering of the functions that are computed in the network is known as the service function chaining problem\cite{chain1,chain2,chain3}. 
The  forwarding graph embedding problem is to assign each VNF to one node,  based on the set of  requests. This problem is known to be NP-hard \cite{nfvra}. Some architectures are proposed to address  this issue \cite{fgraph1,fgraph2, fgraph3}.
 The final stage in the NFV-RA  is the problem of scheduling,  where we want to  assign the execution of a given VNF to one  node of the network, in addition to setting the  execution times of those functions at the network  nodes, given the set of the requests.    This problem is formulated in a basic form in \cite{sch1,sch2}, and some solutions are proposed in \cite{sch3}.  
 Joint NFV-RA and admission control problem have been recently proposed in the literature to increase the resource efficiency, e.g., \cite{MAT}.

 In this paper,  we deal with the same problem of resource allocation in NFV-based networks, with the  objective of minimizing the end-to-end transmission delay. However, we address the problem through a completely novel analytic approach. 
 The key contribution of this work is to demonstrate that under certain conditions, it is possible to reduce the end-to-end delay for computation of VNFs over the network  if functions are allowed to be executed redundantly and in parallel over a number of nodes, an assumption that has not been properly exploited earlier in the literature.
 Another key observation in this paper is that by allowing a function to be computed redundantly over a number of nodes, a trade-off between overall computational power and end-to-end delay is observed. 
 To the best of our knowledge, 
 this is the first time that such trade-off is investigated in VNF computation.
 To achieve this trade-off, we rely on two components: (1) redundancy in computation as mentioned above, (2) network coding in delivery. 

It is interesting to note that, recently, the role of coding for achieving fundamental trade-offs between computation and communication has been investigated in other problems \cite{code1,code2,code3,code4, code5,code6,opt1,opt2,opt3,gcode1,gcode2,gcode3}.  
In \cite{code1}, the authors consider the framework of Map-Reduce, which is a known framework for distributed computing, and  show that  coding techniques can improve the performance of the system significantly. 
Such improvements in  performance measures via coding theory  are also observed in wireless distributed computing  problem \cite{code4}.
The problem of network stragglers is also addressed through use of coding theory for computational tasks.
For example, in the problems of distributed matrix multiplication \cite{code2} and coded Fourier transform \cite{code3},  the use of coding theory for computational tasks can alleviate the problem of network stragglers. 
Also in the problems of distributed optimization, where  an optimization problem is divided into  sub-optimization problems to reduce the computational complexity and  decrease  the delay using parallel servers,  the problem of straggler servers is  challenging. To deal with this challenge, recently,  encoded distributed  optimization is introduced \cite{opt1,opt2,opt3}.
Another example is to use the coding techniques in distributed gradient descent algorithm, which is an important algorithm  in many problems of  machine learning and data sciences\cite{gcode1,gcode2,gcode3}. The benefit of  gradient coding  is also for avoiding the stragglers.
 Also in the NFV-based systems,  the coding theoretic approaches  are recently used in \cite{code5,code6} to reduce the effect of stragglers in VNF computing nodes.
In this work, we introduce network coding to NFV-based networks to achieve optimum transmission delay.

 For finding the minimum end-to-end transmission delay analytically, we propose an auxiliary multicast problem, which captures all  the assumptions we have made before.
 This auxiliary problem shows that, we can  optimize the end-to-end delay by applying the single-source multicast  theorem \cite{code7} to the model considered in this paper.  Thus, the problem of minimization of the end-to-end transmission delay can be formulated mathematically using the capacity of the auxiliary multicast problem, which is known to be equal to the min-cut in the network coding terminology. In this way, we  formulate the problem of  finding the subsets of  nodes that must compute  functions  in each round  by an optimization problem,  specifically as an integer    programming problem.
 Following the results of network coding theory, it is shown that the minimum end-to-end delay can be achieved only via  linear codes \cite{code10,code11}. Actually, there are  polynomial time algorithms for network code construction to exploit the capacity of the system   \cite{code8}. 
 We also notice that in our model,  we do not assume that the functions have any structure, for example they are linear.
Our results hold for  arbitrary functions. 
 
To use these results in practice, as we stated before, a complex  integer    programming problem  needs to be solved. 
This problem corresponds to  obtaining  the set of nodes that will compute any specific function  in  parallel.
 While  finding the exact solution of this  optimization problem has high computational complexity in general, we propose an approximation algorithm for finding the solution with a moderate polynomial complexity. 
We show that  the proposed algorithm outperforms the traditional no-redundancy approaches via numerical simulations.
Furthermore, if we have an  upper bound of $\alpha$  on the number of nodes that compute the given function at each round of function computation, where $\alpha$ does not scale with the size of the network, an optimal algorithm is also provided  with polynomial complexity. This algorithm is based on the dynamic programming.  

At the final stage, we present some numerical simulations and show the gain of the proposed scheme in several  cases, specially  in comparison with  no-redundancy scheme. 
This is observed, both theoretically and numerically, that we have a trade-off between the end-to-end delay and the  processing cost which means  the maximum number of nodes that can compute a given function in parallel.
It is interesting to note that, in fact, there are some scenarios that by increasing the number of nodes that can compute a given function in parallel by a multiplicative factor of $\alpha$, the end-to-end delay will also decrease by the same factor.
To the best of our knowledge, it is the first time in the literature that such trade-off is addressed.
Note that  a natural question about our method is whether  it is possible to bound the gain of redundant computing by a constant multiplicative factor for all networks. 
As a contribution, we have developed some examples in which the gain of redundant computing concatenated with network coding scales with the size of the network. Hence, such constant does not exist.

The rest of this paper is organized as follows. In Section \ref{ps}, we define the problem mathematically. In Section III, we state the main result. The proposed algorithms can be found in Section IV.  We present the numerical simulations in Section V and finally, Section VI concludes the paper.

\textbf{Notation.} For any positive integer $K$, we define $[K]:=\{1,2,\ldots,K\}.$ Also, vectors are denoted by bold letters, like $\bold{x}.$

\section{Problem Statement}\label{ps}

\begin{figure*}
\centering
\includegraphics[scale  = 0.32]{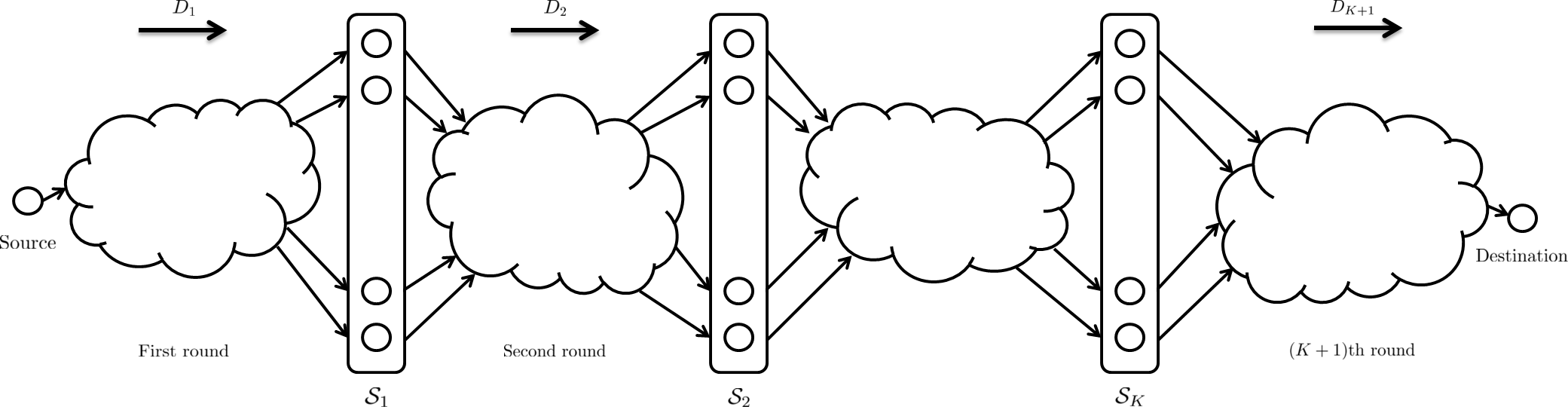}
\caption{A responding scheme. There are $K+1$ rounds of data transmission, sequentially from one set to the next. In any set $\mathcal{S}_k$, a similar function computing task is performed in  all of the nodes of the set, simultaneously and in parallel. 
Also in any transmission round $k$, there exists the infrastructure graph $\mathcal{G}$ such that a common data must be delivered from the set $\mathcal{S}_{k-1}$ to $\mathcal{S}_{k}$ using it. The networks, which are shown by clouds in the figure,  are $K+1$ copies of the infrastructure graph.
}
\end{figure*}

Consider a communication network modeled by a  directed 
 weighted  graph. We represent the set of   nodes in this network  by     $\mathcal{V}$ and  the set of  directed links by $\mathcal{E}\subseteq \mathcal{V}\times \mathcal{V}$. 
For any directed link $e=(u,v) \in \mathcal{E}$, its communication capacity  is denoted by $w(e)$. This means that it is possible to send a packet of length $\lfloor T w(e) \rfloor$ bits from $v$ to $u$ via this  link, during $T$ time slots and assuming that the links are error free. Let us denote this network by $\mathcal{G}=(\mathcal{V},\mathcal{E},w(.))$, which  corresponds to the network function virtualization infrastructure (NFVI) in the literature.

There is a library of  functions denoted by  $\mathcal{F}$ which indicates the set of  VNFs in our model. We  assume that  each function\footnote{In this paper, we utilize VNF and function interchangeably.} 
$f\in \mathcal{F}$ is a deterministic mapping from the set $\{0,1\}^{U_{f}}$ to the set $\{0,1\}^{L_{f}}$, for some positive integers $U_{f}$ and $L_{f}$, respectively.  
 We denote  the set of  nodes that have ability to compute a function 
(or equivalently a VNF) $f \in \mathcal{F}$ by $\mathcal{V}_{f}\subseteq \mathcal{V}$  for any $f \in \mathcal{F}$. 
 This means that in the placement phase of VNFs, we place any function $f$ into a subset of the network nodes, denoted by $\mathcal{V}_{f}\subseteq \mathcal{V}$. 
In this setup, we assume that $\mathcal{V}_f$ is predetermined and 
 the run-times of a function (or VNF) at different nodes are not different, i.e., the processors are  assumed to be homogeneous.
 
In our model, we does not put any restriction to the number of functions that a node can compute.
 However,  for  simplicity, we assume that any node can compute at most one function, i.e., $\mathcal{V}_f\cap \mathcal{V}_g = \emptyset$  for any distinct functions $f$ and $g$. This assumption does not affect the generality of the problem.  
Note that,  for any  node $v \in \mathcal{V}$ that can compute $k>1$  distinct functions (VNFs), we can replace it by $k$ new nodes $\tilde{v}_{1}, \tilde{v}_{2},\ldots, \tilde{v}_{k}$ that can run exactly one of the $k$ functions and form a new directed weighted graph as follows. All  the links from or to $v$ in the  previous network are considered from or to $\tilde{v}_{1}$ in the new network. Also, for any distinct  $k_1,k_2\in  [K]$,  we connect $\tilde{v}_{k_1}$ to $\tilde{v}_{k_2}$ by a directed link with infinite capacity, i.e.,  $w(\tilde{v}_{k_1},\tilde{v}_{k_2})=\infty$. The other nodes, links and weights   are remained unchanged. 
With this modification, the resulting network and the original network are equivalent in terms of delay and the achievable scheme. This shows that the condition is not restrictive.

There is a request tuple  $\mathscr{R}=(s,d,\bold{x},(f_1,f_2\ldots,f_K))$ for the computation in   network which is defined as follows. 
The transmitter node or source node $s\in \mathcal{V}$ has the data packet $\bold{x}\in\{0,1\}^{L_0}$, where $L_0$ is an arbitrary positive integer.  For any $k \in [K]$, $f_{k} \in \mathcal{F}$ is the function that must be computed in the $k^{\text{th}}$ round of the function computation problem. 
The receiver or destination node $d\in \mathcal{V}$  is interested to be delivered  the sequentially computed result  $f_{K}(f_{K-1}(\ldots (f_{1}(\bold{x}))\ldots))$. 
We assume that the chain is predetermined and the order cannot be changed.
Note that the output of the $k^{\text{th}}$  function $f_{k}$ must match in the size with the input of the $(k+1)^{\text{th}}$ function $f_{k+1}$. This shows that  there is a sequence of positive integers $L_0,L_1,L_2,\ldots L_{K}$, such that for any $ k \in  [K]$ we have $f_{k}: \{0,1\}^{L_{k-1}}\rightarrow \{0,1\}^{L_{k}}$. Throughout this paper, we assume that each $L_k$ is large enough. This is essential to establish the results\footnote{
Based on the concepts of network coding theory,
it is necessary for the message size (or equivalently $L_k$) to be large enough, in order to ensure existence of  a capacity achieving network code. See \cite{code8} and references therein for more details. 
}.

The problem is as follows. Given the network graph  $\mathcal{G} = (\mathcal{V},\mathcal{E},w(.))$, the library of functions $\mathcal{F}$, a request tuple  $\mathscr{R}=(s,d,\bold{x},(f_1,f_2\ldots,f_K))$ and the sets $\{\mathcal{V}_{f_k}\}_{k \in [K]}$, how can the network use its resources to response $\mathscr{R}$,  in order to minimize the end-to-end  delay? First  the notions of  \textit{responding to a request} and \textit{end-to-end delay} must be clarified.

Considering a request  $\mathscr{R}=(s,d,\bold{x},(f_1,f_2\ldots,f_K))$, 
in our model, a responding  scheme to this request  consists of two steps. 
In the first step, a sequence of non-empty sets  $(\mathcal{S}_1,\mathcal{S}_2,\ldots,\mathcal{S}_K)$ must be selected, where $\mathcal{S}_k\subseteq \mathcal{V}_{f_k}$ for any $k\in [K]$. For any $k$, $\mathcal{S}_k$ is the set of the nodes that will compute the function $f_k$ in the $k^{\text{th}}$  round of the function computation problem. 
In other words, all of  the nodes in $\mathcal{S}_k$ must compute the function $f_k$. Apparently $|\mathcal{S}_k| \ge 1$. Later, we will see why it may be beneficial to have $|\mathcal{S}_k| > 1$. 
 Also for the simplicity in notation, we define $\mathcal{S}_0=\{s\}$ and  $\mathcal{S}_{K+1}=\{d\}$.

The second step consists of  $K+1$ transmission rounds. 
In the first round,  the  data packet $\bold{x}$ must be delivered from the source node $s$ to all of the   nodes of the set $\mathcal{S}_1$ simultaneously, using the network links. After this round, all of the nodes of the set $\mathcal{S}_1$ compute the function $f_1(\bold{x})$ simultaneously and in parallel manner. 
Now $f_1(\bold{x})$ is available in  all of the nodes of the set $\mathcal{S}_1$.
 In the second round, $f_1(\bold{x})$ must be delivered from the set of nodes $\mathcal{S}_1$  to each node of the set   $\mathcal{S}_2$, separately.
  Therefore, in the $k^{\text{th}}$ round, the nodes of the set  $\mathcal{S}_{k-1}$ have $f_{k-1}(f_{k-2}(\ldots (f_1(\bold{x}))\ldots))$ and they must deliver this common data to each node of the set    $\mathcal{S}_{k}$, using the network links.
   Note that by this definition, the functions must be computed independently and sequentially.  In other words, the order of function computing is important. Also, for any $k \in [K]$,  the value of $f_{k}(f_{k-1}(\ldots (f_1(\bold{x}))\ldots))$ cannot be computed, unless the value of $f_{k-1}(f_{k-2}(\ldots (f_1(\bold{x}))\ldots))$ is available at a node that can compute $f_k$.  

More formally, in the second step,  the sets $(\mathcal{S}_1,\mathcal{S}_2,\ldots,\mathcal{S}_K)$ are available. The data must sequentially be delivered from one set to another. Hence, we assume that a responding scheme must include a sequence $(\mathfrak{D}_1,\mathfrak{D}_2,\ldots,\mathfrak{D}_{K+1})$ of delivery schemes. For any $k \in [K+1]$, $\mathfrak{D}_k$ is a transmission policy of the common content $f_{k-1}(f_{k-2}(\ldots (f_1(\bold{x}))\ldots))$, from the set $\mathcal{S}_{k-1}$ to the set $\mathcal{S}_{k}$. 
Notice that all of the nodes of the set $\mathcal{S}_{k-1}$ have the common data $f_{k-1}(f_{k-2}(\ldots (f_1(\bold{x}))\ldots))$ and they want to deliver this common data to the  nodes of the set $\mathcal{S}_{k}$ where  all of the nodes in $\mathcal{S}_{k}$ must be delivered the content $f_{k-1}(f_{k-2}(\ldots (f_1(\bold{x}))\ldots))$.
In this paper, we consider a multistage graph where in each stage, i.e., between $S_{k-1}$ and $S_k$, there exists the infrastructure network  $\mathcal{G}$   (see Figure 1). Therefore, in Figure 1, there exist $K + 1$ copies of infrastructure network for $K + 1$ stages where each cloud is utilized to demonstrate this copy of $\mathcal{G}$  for each stage.

For a responding scheme which corresponds to $(\mathcal{S}_1,\mathcal{S}_2,\ldots,\mathcal{S}_K)$ and  $(\mathfrak{D}_1,\mathfrak{D}_2,\ldots,\mathfrak{D}_{K+1})$, the \textit{end-to-end  delay} is defined as follows. For any $k\in [K+1]$, assume that the delay of the  transmission round $k$ is denoted by $D_k$ which depends on the request and the responding scheme. More precisely, $D_k$  is a function of the sets $\mathcal{S}_{k-1}$ and $\mathcal{S}_{k}$,  and  the delivery scheme $\mathfrak{D}_k$. The end-to-end transmission delay $D$  is defined as the summation  of the transmission round delays, i.e., $D:=\sum_{k=1}^{K+1}D_k$.
In this paper, our objective is to minimize $D$, subject to the network considerations and the request. 
The design parameters are the sets $\{\mathcal{S}_k\}_{k\in [K]}$ and  the delivery schemes $\{\mathfrak{D}_k\}_{k \in [K+1]}$.

 Note  that in the traditional responding schemes, any  set $\mathcal{S}_{k}$    is a one-element set which means that   $f_k$ is computed at exactly one node.  We call this scheme as no-redundancy  approach.
  In this paper, we omit this restriction and allow  the responding schemes to compute a function redundantly in several nodes of the network in parallel. 
  We will show that via redundancy in  function computing, the network delay can decrease, compared to the no-redundancy  schemes.
  We note that this approach for responding a request, however, may increase the delay in the first round of  the data transmission, because the content $\bold{x}$ must be delivered to more than one nodes. On the other side, for the second round of data transmission, the content $f_1(\bold{x})$ is available in more than one nodes. 
This means that redundantly function computing may decrease the delay in the second round of data transmission. 
Similarly, this phenomenon is observed in all rounds of computation. 
This shows a trade-off for  choosing one-element sets or larger sets.

\section{Main Results}
In this section, we state the main results of this  paper. In particular, we propose a scheme that minimizes the end-to-end  delay $D$. Furthermore, we prove that the proposed  scheme is optimal, for the objective of delay minimization. First we give some examples to explain the method.

\subsection{Illustrative Examples}

\begin{example} \normalfont

 \begin{figure}
 \centering

 \begin{subfigure}	[b]{1\textwidth}
 \centering

\includegraphics[scale = 0.9]{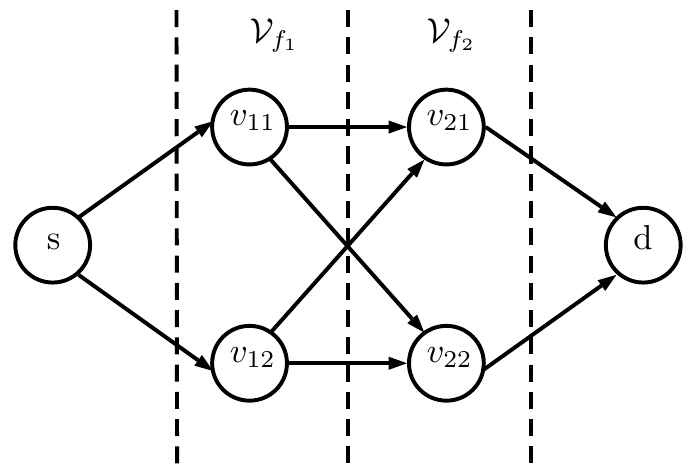}

\caption{The network of Example 1. }

\end{subfigure}

\hfill

\begin{subfigure}[b]{1\textwidth}
 \centering

\includegraphics[scale = 0.9]{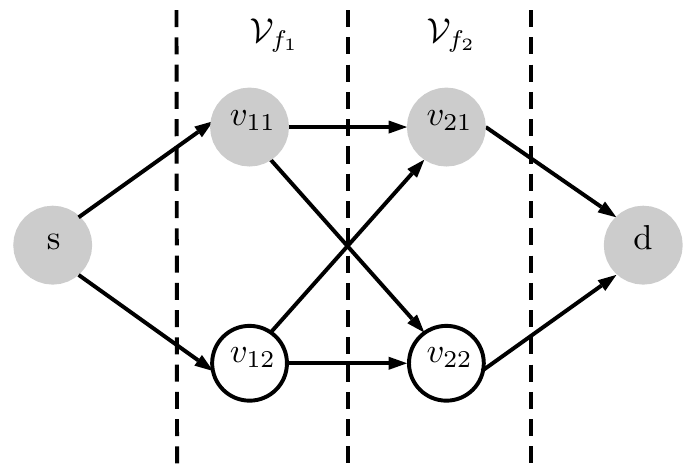}
\caption{The node $s$ transmits  $\bold{x}$  through the directed link to the node $v_{11}$. The node $v_{11}$ is  delivered $\bold{x}$  after $L_0$ time slots, and then computes $f_1(\bold{x}).$ Similarly, $v_{11}$ transmits  $f_1(\bold{x})$ to the node $v_{21}$ using the directed link between them. The node $v_{21}$ receives  $f_1(\bold{x})$ after $L_1$ time slots and then computes $f_2(f_1(\bold{x}))$ and transmits it to the destination $d$. Hence, using one-element sets $\mathcal{S}_1=\{v_{11}\}$ and $\mathcal{S}_2=\{v_{21}\}$, one can achieve the normalized delay of  $D/L_0=3.$ Note that the capacity of each link is set to  one.}

 \end{subfigure}

\hfill

\begin{subfigure}[b]{1\textwidth}

\centering

\includegraphics[scale = 0.9]{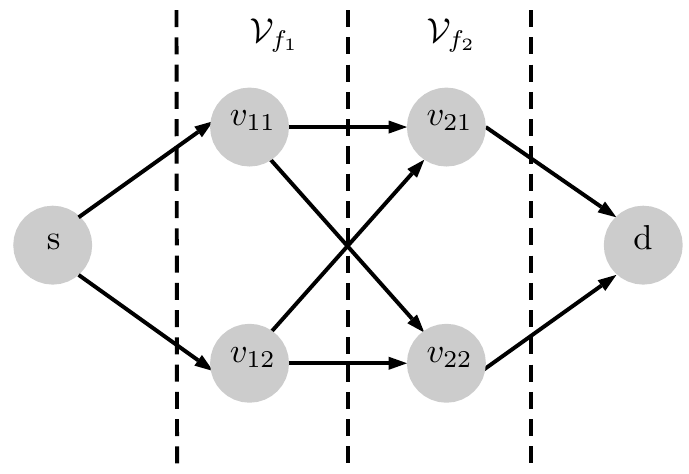}

\caption{ The proposed scheme for  function computation in network. Assume that in the first round, the source node $s$ sends $\bold{x}$ to the  nodes $v_{11}$ and $v_{12}$ simultaneously. This means that after $L_0$ time slots, two nodes $v_{11}$ and $v_{12}$ are delivered the data $\bold{x}$ and can compute $f_1(\bold{x})$.  Now assume that $f_1(\bold{x})=(\bold{y}_1,\bold{y}_2)$, where $\bold{y}_1$  contains the first $L_1/2$ bits of $f_1(\bold{x})$ and $\bold{y}_2$ contains the second $L_1/2$ bits of  $f_1(\bold{x})$. In the second round of the data transmission, the node $v_{11}$ transmits $\bold{y}_1$ through the directed links to $v_{21}$ and $v_{22}$. Simultaneously, the node $v_{12}$ transmits $\bold{y}_2$ through the directed links to the nodes $v_{21}$ and $v_{22}.$ Using this approach, after $L_1/2$ time slots, the nodes $v_{21}$ and $v_{22}$ are delivered  $f_1(\bold{x})$ entirely. This means that by this approach, the second round of data transmission takes only $L_1/2$ time slots. Similarly, one can see that the third round can be performed after $L_2/2$ time slots. This shows that the total delay is $D=L_0+L_1/2+L_2/2=2L_0$, and hence, the normalized delay of $D/L_0 = 2$ is achievable.}

\end{subfigure}
\caption{The network of Example 1 and the proposed method for function computation. }

\end{figure}

Consider the network in Figure 2 (a) and the request tuple  $\mathscr{R}=(s,d,\bold{x},(f_1,f_2))$. We notice that the data  $\bold{x}$ is available in the source node $s$. There are two functions $f_1$ and $f_2$. The destination node $d$ is  interested  in receiving the computed result of $f_2(f_1(\bold{x}))$. We assume that the functions $f_1$ and  $f_2$ can be computed in the nodes $\mathcal{V}_{f_1}=\{v_{11},v_{12}\}$ and $\mathcal{V}_{f_2}=\{v_{21},v_{22}\}$, respectively.
 The capacity of each link is  equal to one, i.e., $w(e)=1$ for any link $e$. Also, we assume that  $L_0=L_1=L_2$. 
 One can see that using  no-redundancy  scheme for   function computation,  the normalized delay of $D/L_0 = 3$ is achievable (see Figure 2 (b)).
  Furthermore, it is obvious that using  only one-element sets for $\mathcal{S}_1$ and $\mathcal{S}_2$, we cannot have an  end-to-end transmission delay  which is  smaller than $3L_0$. On the other side, if we redundantly compute the functions $f_1$ and $f_2$,  we can achieve a more small delay. In particular, consider $\mathcal{S}_1=\{v_{11},v_{12}\}$ and $\mathcal{S}_2=\{v_{21},v_{22}\}$, 
  which means that the functions $f_1$ and    $f_2$ are computed in  two nodes simultaneously and in parallel. 
In this case, we can prove that the total delay decreases  (see Figure 2 (c)). It is shown in Figure 2 (c)  that using this method,  the normalized delay   of $D/L_0 = 2$ is achievable. This means that in this example, redundantly function computation outperforms the traditional no-redundancy  method, by a factor of $3/2$.

\begin{remark} \normalfont
Note that in Example 1 and also in the next example, there is a bipartite graph  between  nodes computing $f_k$ and   nodes computing $f_{k+1}$. 
However, in the rest of this paper, we consider a general scheme where we assume existence of an arbitrary network between nodes in layer $k$ and layer $k+1$
which may not be necessarily one-hop
 (see the clouds in Figure 1). 
In other words, we do not restrict ourselves  to a specific network topology.
\end{remark}

\end{example}

\begin{example} \normalfont
Consider the network in Figure 3 (a) and the request tuple  $\mathscr{R}=(s,d,\bold{x},(f_1,f_2,\ldots,f_K))$.  Similar to Example 1, there is a data packet $\bold{x}$,  available in the source node $s$. The destination node $d$ is interested in receiving the computed result of $f_K(f_{K-1}(\ldots f_1(\bold{x})\ldots)))$.
 Assume that any function $f_k$ can be computed in the nodes $\mathcal{V}_{f_k}=\{v_{k1},v_{k2},\ldots, v_{kN}\}.$ In addition, let $w(e) = 1$ for any link $e$. In this example,  assume that  $L_0=L_1=\ldots=L_K.$ Similar to Example 1,    $D/L_0 = K+1$ is the best achievable normalized delay using  no-redundancy approaches for the function computation (see Figure 3 (b)).  However, one can compute the functions redundantly in the nodes. In particular, assume that to respond the request, any function $f_k$ is computed in  all of the nodes which are capable to compute it, i.e.,  $\mathcal{S}_k=\mathcal{V}_{f_{k}}$. Via this approach, the network delay   decreases. From Figure 3 (c),  the normalized delay of $D/L_0=(1+K/N)$ is achievable by this method.  This shows that using the proposed approach for function computation, the end-to-end delay decreases by a factor of $\frac{1+K}{1+K/N}$. For large values of $K$,   we have  $\frac{1+K}{1+K/N} \approx N$. This means that the  gain of the redundantly function computing  can increase as the size of the network increases.

\begin{figure}
\centering

 \begin{subfigure}	[b]{1\textwidth}
 \centering

\includegraphics[scale = 0.47]{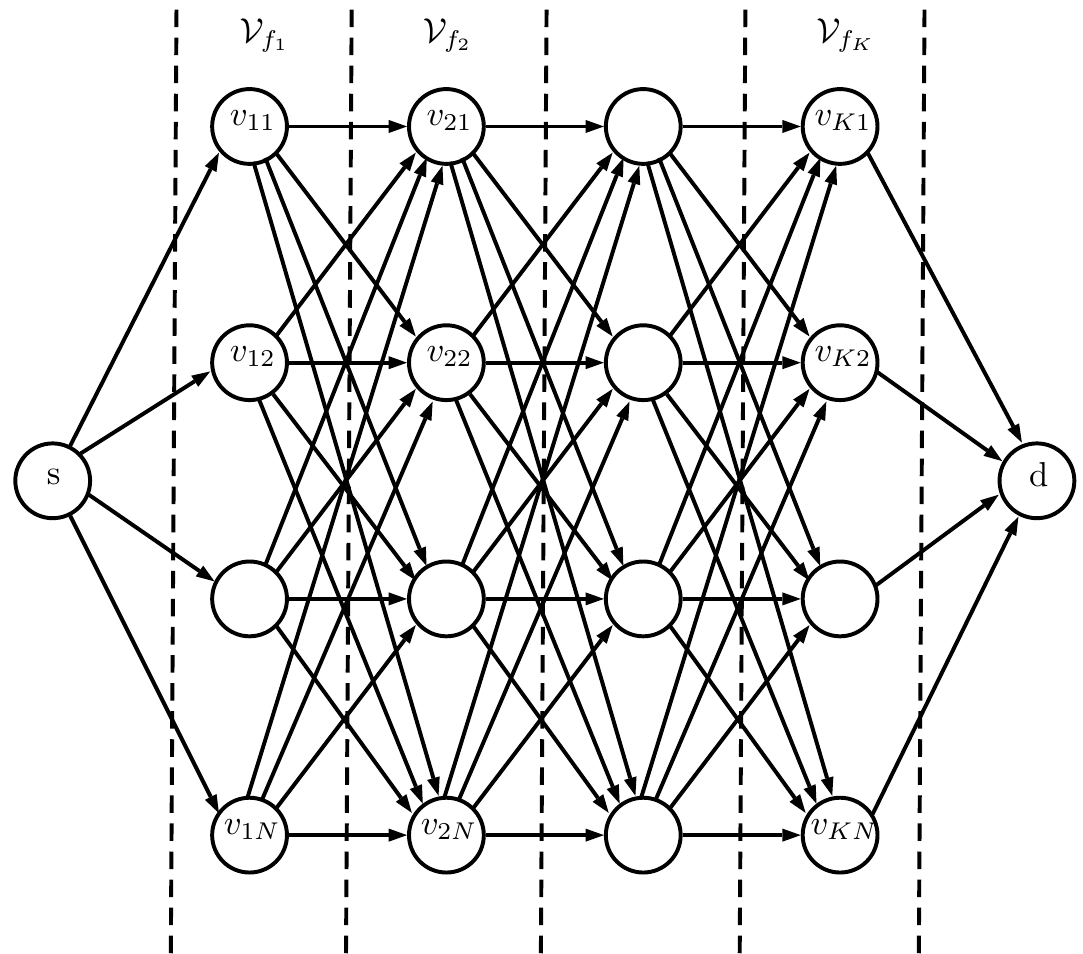}
\caption{The network of Example 2.  For any $n,n' \in [N]$ and any $k \in [K-1]$,   $(v_{kn},v_{(k+1)n'})$ is a directed link in the network. Also, for any $n \in [N]$,  the the source node $s$ is connected to the node $v_{1n}$ and the  the node $v_{Kn}$ is connected to the destination node $d$.  It is assumed that  the capacity of any link is set to  one. }
\end{subfigure}

\hfill

 \begin{subfigure}	[b]{1\textwidth}
 \centering

\includegraphics[scale = 0.47]{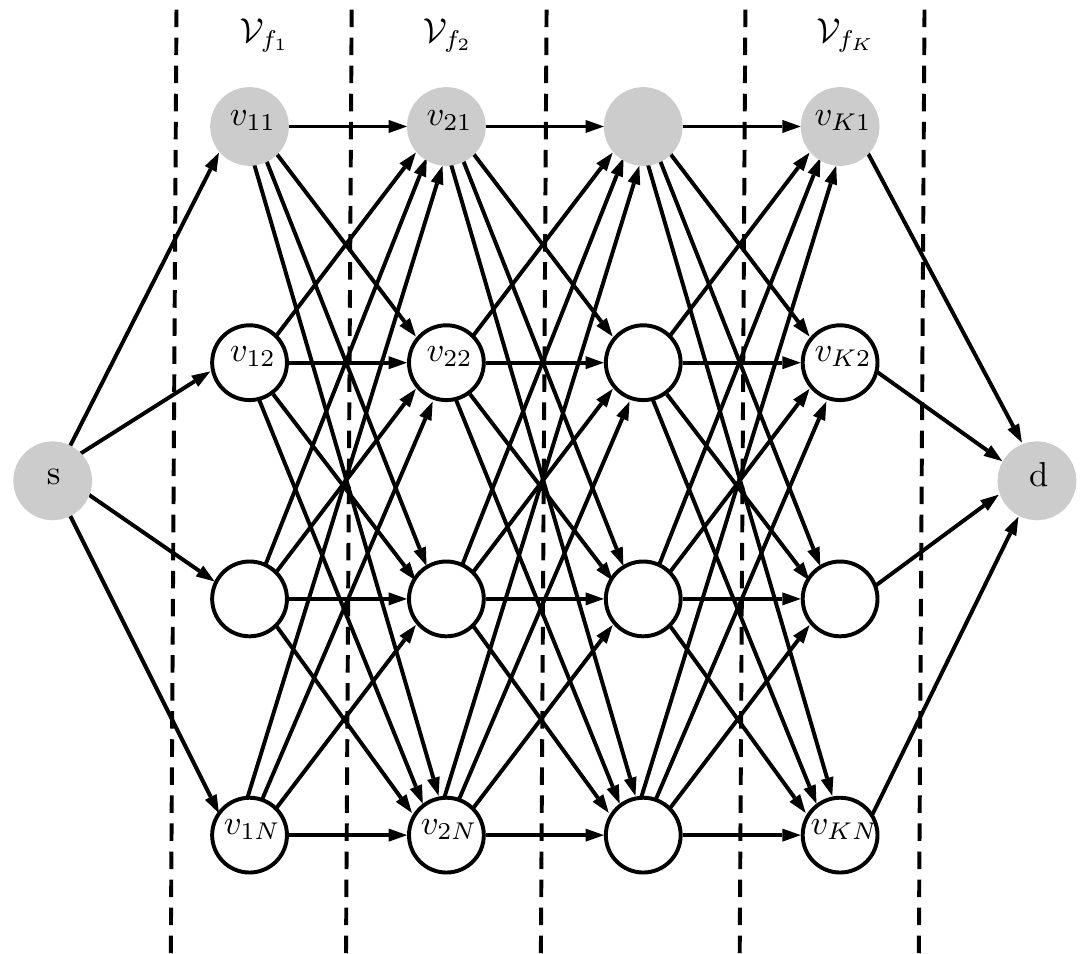}
\caption{The source node $s$ transmits $\bold{x}$ to the node $v_{11}$  in the first round of data transmission. This round has the delay of $L_0$ time slots. In the second round, node $v_{11}$ computes $f_1(\bold{x})$ and sends it to node $v_{21}$ in $L_1$ time slots. Similarly, the computation and transmission  tasks sequentially continue, and finally, the node $d$ is delivered the computed result of $f_K(f_{K-1}(\ldots f_1(\bold{x})\ldots)))$ after $L_0+L_1+\ldots+L_K = L_0(K+1)$ time slots.}
 
\end{subfigure} 

\hfill
 
\begin{subfigure}	[b]{1\textwidth}
\centering
 
\includegraphics[scale = 0.47]{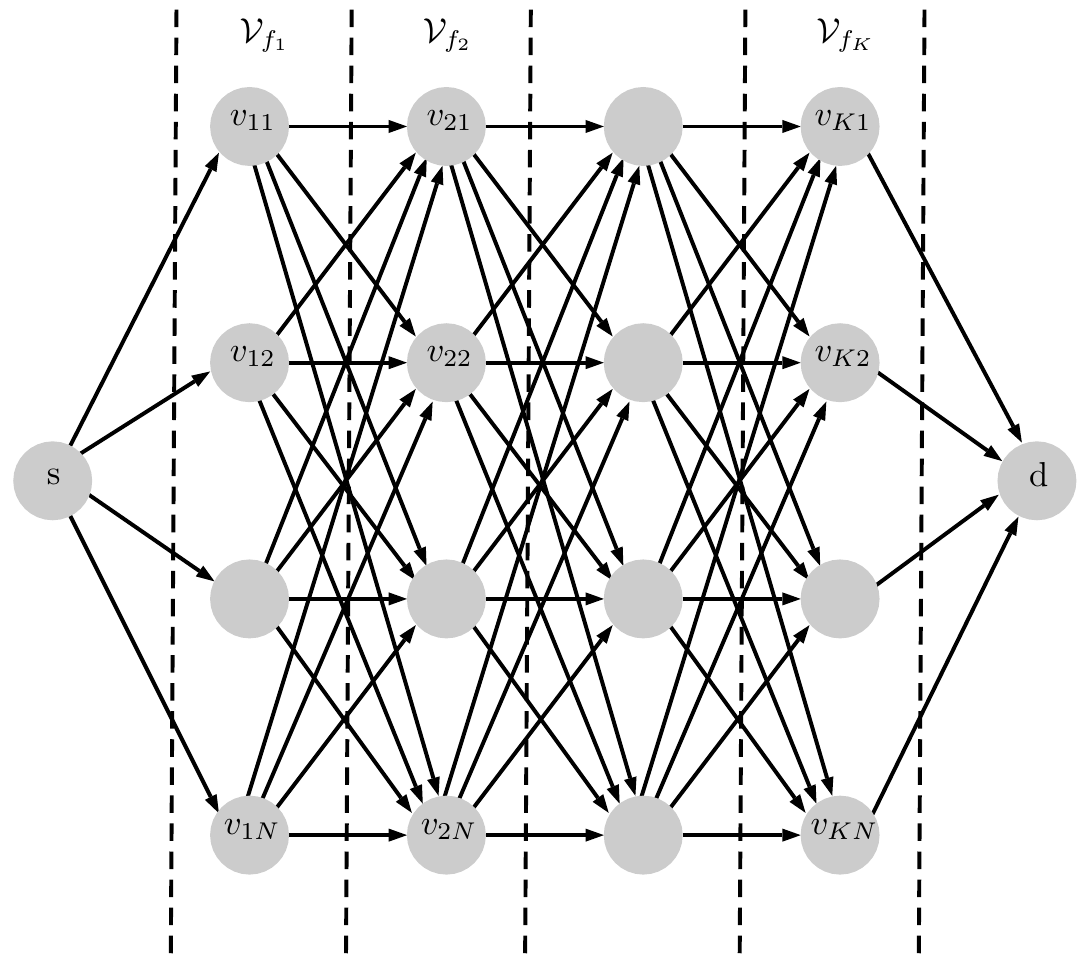}
\caption{The node $s$ transmits the data packet $\bold{x}$  to all of the nodes that can compute $f_1$ in the first round. This round lasts $L_0$ time slots. Next in the second round,  $f_1(\bold{x})$ is available in  all of the nodes $\{v_{11},v_{12},\ldots, v_{1N}\}$. Now divide $f_1(\bold{x})$ into $N$ equal fragments. Assume that each node $v_{1n}$ transmits the $n^{\text{th}}$ fragment of $f_1(\bold{x})$ to all of the nodes that have a directed link from $v_{1n}$. Using this approach  after $L_1/N$ time slots all of the nodes of $\mathcal{S}_2$  are delivered the content $f_1(\bold{x})$ entirely. This means that the second round of transmission lasts $L_1/N$ time slots. Using similar approach for the other rounds, one can see that the end-to-end delay of $L_0+L_1/N+\ldots + L_K/N=L_0(1+K/N)$ is achievable. 
}
 
\end{subfigure}
\caption{The network of Example 2 and the proposed scheme for function computation.}
 
\end{figure}

\end{example}

\subsection{Main Results}

In the previous subsection,  the  main idea of redundant function computation has been discussed. For the statement of the main result for  general networks, first we need a definition.

 \begin{definition} (Cut)
 For a directed weighted graph $\mathcal{G} = (\mathcal{V}, \mathcal{E},w(.))$, and any $\mathcal{S}\subseteq \mathcal{V}$, $\cut(S)$ is defined as 
 $$\cut(\mathcal{S})=\sum_{
\substack{
            e=(u,v)\in \mathcal{E}\\
           u \in \mathcal{S};
            v \notin \mathcal{S}
           }
 }w(e).$$
 \end{definition}

In the next theorem, we state a tight lower bound on $D_k$.

\begin{theorem}
For any $k \in [K+1]$, any  $\mathcal{S}_{k-1}\subseteq \mathcal{V}_{f_{k-1}}$, and any $\mathcal{S}_{k}\subseteq \mathcal{V}_{f_{k}}$, the delay of the $k^{\text{th}}$ transmission round can be lower bounded as follows
\begin{align}
D_k \ge \frac{L_{k-1}}{\min\limits_{
\substack{
            \mathcal{S}_{k-1} \subseteq \mathcal{S} \subseteq \mathcal{V}\\
             \mathcal{S}_{k} \not \subseteq \mathcal{S}}
} \cut(\mathcal{S})}.
\end{align}
Moreover, the lower bound is tight, i.e., there exists a transmission scheme $\mathfrak{D}_k$ that achieves the lower bound.

\end{theorem}
\begin{proof}
See appendix A.
\end{proof}
\begin{remark} \normalfont
To prove Theorem 1, we utilize the results from network coding,  specifically the capacity of the single multicast problem \cite{code7}. We show that the delivery problem in each round can be reduced to a single-multicast problem where the capacity is equal to the min-cut, and can be achieved by linear codes.

\end{remark}

\begin{remark} \normalfont
The above result shows that routing is not generally optimum. In particular, for delivery  network coding should be applied. 

\end{remark}

\begin{corollary} \normalfont
From Theorem 1,  for a given sequence of sets $(\mathcal{S}_1,\mathcal{S}_2,\ldots,\mathcal{S}_K)$, the end-to-end delay is lower bounded by 
\begin{align}
D \ge \sum_{k=1}^{K+1} \frac{L_{k-1}}{\min\limits_{
\substack{
            \mathcal{S}_{k-1} \subseteq \mathcal{S} \subseteq \mathcal{V}\\
             \mathcal{S}_{k} \not \subseteq \mathcal{S}}
}  \cut(\mathcal{S})}
\label{01}.
\end{align}
In addition, the lower bound is achievable.
\end{corollary}

\begin{theorem}
The optimum end-to-end delay of the problem, denoted by $D^*$,  is equal to
\begin{align}
D^* = \min_{\forall k\in[K]: \mathcal{S}_k \subseteq \mathcal{V}_k}  \Big \{  \sum_{k'=1}^{K+1} \frac{L_{k'-1}}{\min\limits_{
\substack{
            \mathcal{S}_{k'-1} \subseteq \mathcal{S} \subseteq \mathcal{V}\\
             \mathcal{S}_{k'} \not \subseteq \mathcal{S}}
} \cut(\mathcal{S})} \Big \} \label{13}.
\end{align}
In addition, the optimum sets $(\mathcal{S}^*_1,\mathcal{S}^*_2,\ldots,\mathcal{S}^*_K)$ are arguments of the above optimization.

\begin{proof}
See appendix B.
\end{proof}

\begin{remark} \normalfont
To achieve $D^*$, we let $\mathcal{S}_1=\mathcal{S}^*_1,\mathcal{S}_2=\mathcal{S}^*_2,\ldots,\mathcal{S}_K=\mathcal{S}^*_K$ and we develop linear codes in delivery rounds.
\end{remark}

\end{theorem}

\begin{remark} \normalfont
Via Theorem 2, the minimum achievable end-to-end delay of the problem for a general network is characterized. However,  to achieve this minimum delay, a complicated optimization problem must be solved to obtain the sets $\{\mathcal{S}_k\}_{k \in [K]}$.
Solving this optimization problem  is not straight-forward. In the next section, we  propose some approximation algorithms to solve this problem.
\end{remark}

We note that the  size of the set $\mathcal{S}_k$  corresponds to the processing cost of  $f_k$. Hence, it is rational to have  constraint on the cardinality of each set $\mathcal{S}_k$. The following corollary states the optimization problem with the processing power constraint on the sets. 
This means that we restrict the optimization to the cases that each set  $\mathcal{S}_k$ has a bounded size.

\begin{corollary}  \normalfont
The problem of minimization of the end-to-end delay, subject to limit on processing power in each round is formulated as 
\begin{align}
\min_{
\substack{
                         \forall k\in[K]: \mathcal{S}_k \subseteq \mathcal{V}_k \\
             |\mathcal{S}_k | \le \alpha_k}
}  \Big \{  \sum_{k'=1}^{K+1} \frac{L_{k'-1}}{\min\limits_{
\substack{
            \mathcal{S}_{k'-1} \subseteq \mathcal{S} \subseteq \mathcal{V}\\
             \mathcal{S}_{k'} \not \subseteq \mathcal{S}}
} \cut(\mathcal{S})} \Big \},
\end{align}
where $\alpha_k$ denotes the maximum processing cost of $f_k$.
\end{corollary}

\begin{remark}\normalfont
Note that the above formula for the end-to-end delay shows that by increasing the processing costs, i.e., $\alpha_k$'s, the end-to-end delay can decrease, because the set of feasible solutions of the optimization enlarges.
This shows that there is a trade-off between end-to-end delay and processing cost of each request.
\end{remark}

For the rest of this section, we  explain the results of redundantly function computation in a special case of complete graphs. 
The main goal of this example is to demonstrate that there are also fully connected networks (not necessarily one-hop like  Example 2) that the gain of the proposed scheme, i.e., redundantly computation and network coding in delivery,  scales with the size of the network.  Also, we attend to  demonstrate the  achieved trade-offs between computation and delay via this example.

\subsection{Another Example}
Consider the network $\mathcal{G} = (\mathcal{V},\mathcal{E},w(.))$ where  $\mathcal{V}$ is a finite set of nodes  and $\mathcal{E}$ contains all of the distinct pairs of the nodes, i.e., the graph is fully connected. In this network, assume that there exists a request of  computation as $\mathscr{R}=(s,d,\bold{x},(f_1,f_2,\ldots,f_K))$, and  for any $k\in [K]$, $\mathcal{V}_{f_{k}} = \{v_{k1},v_{k2},\ldots,v_{kN}\}$ is a set containing $N$ nodes for a positive integer $N$. This means that the network consists of $K\times N$ function computing nodes, in addition of other nodes as relays. Suppose that we have at least $N+1$ relays. 
Also, two nodes $s$ and $d$ 
cannot compute any function.

In this case, for any $e = (u,v) \in \mathcal{E}$ we define
$$
w(e) = 
\begin{cases}
\epsilon ~~~ \text{if } ~u \in \mathcal{V}_{f_k} ~\text{for some} ~k \in [K],~\text{or}~u=s,\\
1~~~\text{otherwise},
\end{cases}
$$
where $\epsilon$ is a small enough positive\footnote{
The need for  $\epsilon$ to be small enough is due  to 
such requirement for the proof of Lemma 1. 
}.
 In this fully connected network,  
let us assume that there is an upper bound on processing cost  as $|\mathcal{S}_k| \le \alpha$ where $\alpha \le N$. 
 Now we apply  results of Theorem 2 and Corollary 2   to this network.

We want to  compute the lower bound in Corollary 2. Note that we have
\begin{align}
\cut(\mathcal{S}_k) =  \epsilon \times |\mathcal{S}_k|\times (|\mathcal{V}|-|\mathcal{S}_k|),
\end{align}
for any $k \in [K]\cup \{0\}$. 
To derive $D^*$ from Theorem 2, in this setup, we introduce the following lemma. 

\begin{lemma}
Assuming small enough  $\epsilon$,  for the mentioned network we have
\begin{align}
\min_{
\substack{
            \mathcal{S}_{k-1} \subseteq \mathcal{S} \subseteq \mathcal{V}\\
             \mathcal{S}_{k} \not \subseteq \mathcal{S}}
}\cut(\mathcal{S}) = \epsilon \times |\mathcal{S}_{k-1}|\times (|\mathcal{V}|-|\mathcal{S}_{k-1}|),
\end{align}
for any $k \in [K+1]$.
\end{lemma}
\begin{proof}
See appendix C.
\end{proof}

Now based on Theorem 1, we write (1) as 
\begin{align}
D_k &\ge \frac{L_{k-1}}{\min\limits_{
\substack{
            \mathcal{S}_{k-1} \subseteq \mathcal{S} \subseteq \mathcal{V}\\
             \mathcal{S}_{k} \not \subseteq \mathcal{S}}
 }  \cut(\mathcal{S})}\\
& \overset{(a)}{=} \frac{L_{k-1}}{ \epsilon \times |\mathcal{S}_{k-1}|\times (|\mathcal{V}|-|\mathcal{S}_{k-1}|)},
\end{align}
where (a) follows by Lemma 1. Therefore
\begin{align}
D \ge \sum_{k=1}^{K+1} \frac{L_{k-1}}{\epsilon \times |\mathcal{S}_{k-1}|\times (|\mathcal{V}|-|\mathcal{S}_{k-1}|)}.
\end{align}
The  R.H.S. of  (9) is minimized when $\mathcal{S}_{k}$ is an arbitrary  subset of $\mathcal{V}_{f_{k}}$ with size $\alpha$ for any $k\in [K]$\footnote{
This is due to the fact that the function $f(x) = x(|\mathcal{V}| - x)$ for $x \in \{1,2,\ldots,\alpha \}$ is increasing since $\alpha \le |\mathcal{V}|/2$.
}. Hence, we have
\begin{align}
D_{\text{Optimum}}\times \epsilon =  \frac{L_0}{|\mathcal{V}|-1}+ \frac{1}{ \alpha \times (|\mathcal{V}|- \alpha)}\sum_{k=1}^{K} L_{k}.
\end{align}
Consider $L_0=L_1=\ldots=L_K$. In this case, we conclude that by redundantly function computation, one can achieve the normalized  delay of $D_{\text{Optimum}}\times \epsilon / L_0 = \frac{1}{|\mathcal{V}|-1}+ \frac{K}{\alpha \times (|\mathcal{V}|-\alpha)}$. Also,  if we utilize no-redundancy  approach, we achieve the normalized delay of $D_{\text{No-redundancy}}\times \epsilon / L_0 = \frac{K+1}{|\mathcal{V}|-1}$, where we have $|\mathcal{S}_k|= 1$ for any $k\in [K]$ in this approach
\footnote{
This fact is due to  equation (9).}. If   $K$ is large, and also $\alpha \ll |\mathcal{V}|$, then we conclude that the improvement is approximately equal to $\alpha$, i.e., 
$$\frac{D_{\text{No-redundancy}}}{D_{\text{Optimum}}} \approx \alpha.$$

We notice that while the parameter $\alpha$   indicates  the available processing  cost of each function, in this example, it can  simultaneously  model the reliability measure. In other words, if we set $\alpha$  to be larger, then any function is computed in more nodes of the network and hence, the system is more  reliable against the node failures in functions computation. This means that as   $\alpha$ increases, the system becomes more reliable. 

Now we notice that for  this example, the fundamental trade-offs among end-to-end delay, reliability and processing cost are obtained. 
In Figure 4, we plot the normalized delay in two cases; the no-redundancy approach and the proposed method. Here, an order-wise improvement in the normalized end-to-end delay is observed. 
Note that  in this plot, $\alpha$ represents the processing cost and simultaneously it models the reliability in the function computation.

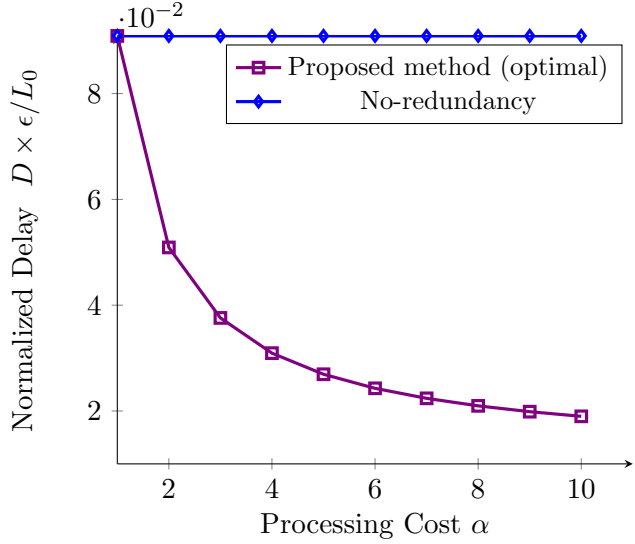
\begin{figure}
\centering
\begin{tikzpicture}
\begin{axis}[
    axis lines = left, 
    xlabel = {Processing Cost~$\alpha$},
    ylabel = {Normalized Delay ~$D \times \epsilon/ L_0$},
     every axis plot/.append style={very thick},
     xmin = 1,xmax = 11,
     ymin = 1/100,ymax = 1/11
]
\addplot [
    domain=1:10, 
    samples=10, 
    color=violet,
        mark=square,
]
{1/99+8/(x*(100-x)) };
\addlegendentry{\small{Proposed method (optimal)}}

\addplot [
    domain=1:10, 
    samples=10, 
    color=blue,
     mark=diamond,
]
{1/11 };
\addlegendentry{\small{No-redundancy}}
\end{axis}
\end{tikzpicture}
\caption{Fundamental trade-off between delay and processing power (or  equivalently reliability) in complete graph. Here $|\mathcal{V}| = 100$, $N = 10$ and $K = 8$.}
\end{figure}

\section{Algorithms}
In this section, we propose our  algorithms for solving the optimization problems (3) and (4).
First we propose an approximation algorithm for solving (3), called greedy algorithm. 
Then, we propose an algorithm for solving (4) in polynomial times. This algorithm is called $\alpha-$optimal algorithm. 
The third algorithm is the same as the greedy algorithm, except  it aims  to solve (4).
We note that all the proposed algorithms are based on dynamic programming.

\subsection{Greedy Algorithm}

\begin{algorithm}[t]
\caption{Greedy  Algorithm}
\begin{algorithmic}[1]
\Procedure{Greedy}{$\mathcal{G}, s,d,\{L_k\}_{k=0}^K \{ \mathcal{V}_{f_k}\}_{k \in [K]}$}
\For{$k \in [K+1]\cup \{0\}$}
\State $\mathcal{P}_k \gets \emptyset$
\State $\mathfrak{B}_k\gets \{   \mathcal{S} \subseteq \mathcal{V}_{f_k} :   |\mathcal{S}| = 1\}$
\EndFor
\State $D \gets \infty$
\While {$1$}
\State $OPT(\{s\})\gets \{s\}$
\State $C(\{s\})\gets 0$
\For{$k \in [K+1]$}
\For{$\mathcal{S} \in \mathfrak{B}_k$}
\State $\hat{\mathcal{S}} \gets  \text{argmin}_{{\mathcal{T}}\in \mathfrak{B}_{k-1} } \{ C({\mathcal{T}}) + \frac{L_{k-1}}{\text{mincut}({\mathcal{T}};\mathcal{S})}\}$
\State $C(\mathcal{S}) \gets C(\hat{\mathcal{S}}) +\frac{L_{k-1}}{\text{mincut}(\hat{\mathcal{S}};\mathcal{S})}$
\State $OPT(\mathcal{S}) \gets (OPT(\hat{\mathcal{S}}),\mathcal{S})$
\EndFor
\EndFor
\If {$D \le C(\{d\})$} 
\State \textbf{break}
\EndIf
\State $ (\mathcal{P}_0,\mathcal{P}_1,\ldots, \mathcal{P}_{K+1})\gets OPT(\{d\})$ 
\State $D \gets C(\{d\})$ 
\For{$k \in [K+1]\cup \{0\}$}
\State $\mathfrak{B}_k\gets \{   \mathcal{S}\cup \mathcal{P}_k: \mathcal{S} \subseteq \mathcal{V}_{f_k} ;  |\mathcal{S}| = 1\}$
\EndFor
\EndWhile
\State \textbf{return} $OPT(\{d\})$ 
\EndProcedure
\end{algorithmic}
\end{algorithm}

In this part, we propose our algorithm for solving the following optimization problem

\begin{align}
D^* = \min_{\forall k\in[K]: \mathcal{S}_k \subseteq \mathcal{V}_k}  \Big \{  \sum_{k'=1}^{K+1} \frac{L_{k'-1}}{\min_{
\substack{
            \mathcal{S}_{k'-1} \subseteq \mathcal{S} \subseteq \mathcal{V}\\
             \mathcal{S}_{k'} \not \subseteq \mathcal{S}}
} \cut(\mathcal{S})} \Big \} \label{13}.
\end{align}
In this optimization problem, our aim is to  find the sets 
$\{ \mathcal{S}_k\}_{k \in [K]}$ in order to minimize the end-to-end delay.
First we note that  the optimization problem
\begin{align}
\text{mincut}(\mathcal{S}_{k-1};\mathcal{S}_k):=\min_{
\substack{
            \mathcal{S}_{k-1} \subseteq \mathcal{S} \subseteq \mathcal{V}\\
             \mathcal{S}_{k} \not \subseteq \mathcal{S}}
} \cut(\mathcal{S}),
\end{align}
is known as  \textit{Max-flow} problem. This problem can be solved via efficient known algorithms such as Ford-Fulkerson algorithm \cite{alg}. Therefore, the optimization problem (11) can be rewritten as
\begin{align}
\min_{\forall k\in[K]: \mathcal{S}_k \subseteq \mathcal{V}_k }  \Big \{  \sum_{k'=1}^{K+1} \frac{L_{k'-1}}{\text{mincut}(\mathcal{S}_{k'-1};\mathcal{S}_k')} \Big \} \label{12}.
\end{align}

Let us define
\begin{align}
C(\mathcal{S}_k):= \min_{\forall k'\in[k-1]: \mathcal{S}_{k'} \subseteq \mathcal{V}_{f_{k'}} }  \Big \{  \sum_{k''=1}^{k} \frac{L_{k''-1}}{\text{mincut}(\mathcal{S}_{k''-1};\mathcal{S}_{k''})} \Big \},
\end{align}
for any $\mathcal{S}_k \subseteq \mathcal{V}_{f_k}$ and any $k\in [K]$. 
Considering  $C( \{s\}) = 0$ and $D^* = C(\mathcal{S}_{K+1})=C(\{d\})$ from (14), we achieve the following recursive equation
\begin{align}
C(\mathcal{S}_k) = \min_{\mathcal{S}_{k-1} \subseteq \mathcal{V}_{f_{k-1}} } \Big \{   C(\mathcal{S}_{k-1}) + \frac{L_{k-1}}{\text{mincut}(\mathcal{S}_{k-1};\mathcal{S}_{k})}\Big \},
\end{align}
for any $k \in [K+1]$. 
The recursive equation (15) shows that the problem of finding the optimum subsets is in the form of dynamic programming and can be solved recursively \cite[Chapter 6]{algkl}.

Now for estimation of the solution of  (\ref{12}), first we use dynamic programming, according to the above discussions,  to obtain one-element sets with minimum delay.
In other words,   we initiate the algorithm by choosing optimum one-element sets.
Then, we enlarge the sets in a greedy manner as long as the overall delay decreases.
 More precisely, at each round, we choose one function computing node from all  the function computing nodes which have not been selected yet, and then, add it to the corresponding subset when this change decreases the end-to-end delay. 
 This procedure is performed, again, via dynamic programming (see lines 8-16 of Algorithm 1) and ends when there is not any function computing  node that adding it to the subsets decreases the delay (see line 17 of Algorithm 1). 
 See Algorithm 1 for the details.

Here,  we provide more explanation about  Algorithm 1.
The parameter $D$ shows the end-to-end delay calculated for each round of greedy algorithm.
 $OPT(\{d\})$ denotes the optimum subsets, chosen from $\mathfrak{B}_k$'s, $k\in [K+1]\cup \{0\}$, resulting the minimum delay. The sets $\mathfrak{B}_k$ for $k\in [K+1]\cup \{0\}$, at first, include  one-element subsets, and then enlarge in a greedy manner via Algorithm 1. 
 Finally at the round that the algorithm does not have any improvement (line 17), the program terminates.
 The final output of the algorithm is the resulting subsets in the final round.

  The complexity of the greedy algorithm is $A \times O((K+1)|\mathcal{V}|^{3})$ where $A$ is the complexity of the algorithm which is used to obtain the value of min-cuts. 
  This can be briefly explained as follows. 
  At each round of dynamic programming based algorithm (see lines 8-16 of Algorithm 1), the function $C(.)$ must be computed for each layer $k$, $k \in [K+1]$. This is due to the fact that,  finally, we want to compute $C(\{d\})$. 
  This procedure ends after  at most $(K+1) \times |\mathcal{V}| \times |\mathcal{V}|$ times of min-cut computation (see lines 10-16 of Algorithm 1). 
  In addition, due to the greedy  addition of nodes, the aforementioned procedure may be repeated at most $ |\mathcal{V}|$ times in the algorithm. Hence, the complexity of the algorithm is $A \times O((K+1)|\mathcal{V}|^{3})$.

\subsection{$\alpha-$optimal Algorithm}

In this part, we aim to solve the following optimization problem
\begin{align}
\min_{
\substack{
                         \forall k\in[K]: \mathcal{S}_k \subseteq \mathcal{V}_k \\
             |\mathcal{S}_k | \le \alpha_k}
}  \Big \{  \sum_{k'=1}^{K+1} \frac{L_{k'-1}}{\min_{
\substack{
            \mathcal{S}_{k'-1} \subseteq \mathcal{S} \subseteq \mathcal{V}\\
             \mathcal{S}_{k'} \not \subseteq \mathcal{S}}
} \cut(\mathcal{S})} \Big \}.
\end{align}

Without loss of generality, assume that $\alpha_k=\alpha$ for any $k$. This is just for simplicity in notation and the proposed algorithm can be considered for the general case.

To solve (16), we propose an algorithm in this part, which is  called  $\alpha-$optimal algorithm. This algorithm is based on dynamic programming to obtain the sets $\{\mathcal{S}_k\}_{k\in[K]}$. 
The algorithm searches among all  the subsets of  $\mathcal{V}_{f_{k}}$ with at most $\alpha$ elements, for any $k$, and chooses the best set to minimize the end-to-end delay, according to the recursive equation of dynamic programming. 
Similar to (15),  the recursive equation of (16) is given by
\begin{align}
C(\mathcal{S}_k) = \min_{
\substack{
\mathcal{S}_{k-1} \subseteq \mathcal{V}_{f_{k-1}}\\
 |\mathcal{S}_{k-1}| \le \alpha}
 } \Big \{   C(\mathcal{S}_{k-1}) + \frac{L_{k-1}}{\text{mincut}(\mathcal{S}_{k-1};\mathcal{S}_{k})}\Big \}.
\end{align}
The $\alpha-$optimal algorithm  computes the function $C(.)$ for each set with at most $\alpha$ elements by  dynamic programming. 
The details can be found in Algorithm 2.
The notations used in this algorithm are also the same as Algorithm 1.

We note that the complexity of the $\alpha$-optimal algorithm is $A\times O((K+1)|\mathcal{V}|^{2\alpha})$ where $A$ is the complexity of the algorithm which is developed to find the min-cuts. 
The remarkable feature of Algorithm 2 is that it is an optimal algorithm with polynomial complexity of the network size.
Let us provide more explanation about the complexity of this algorithm.

In this algorithm, we use dynamic programming only once. 
In that procedure, the function $C(.)$ must be computed for each layer $k$,  and for each possible $\mathcal{S}_k$, where $k\in [K+1]$. 
We note that the number of  subsets of $\mathcal{V}_{f_{k}}$ with at most $\alpha$ elements, for any $k$, is $O(|\mathcal{V}|^{\alpha})$. 
Hence, by taking lines 7-13 of Algorithm 2 into consideration, we conclude that it is necessary  to compute min-cuts  for at most $O ((K+1)\times |\mathcal{V}|^{\alpha} \times |\mathcal{V}|^{\alpha})$ times. 
This shows that the complexity of the proposed algorithm is  $A\times O((K+1)|\mathcal{V}|^{2\alpha})$.

\begin{algorithm}[t]
\caption{$\alpha$-optimal Algorithm} 
\begin{algorithmic}[1]
\Procedure{$\alpha$-optimal}{$\mathcal{G}, s,d,\{L_k\}_{k=0}^K \{ \mathcal{V}_{f_k}\}_{k \in [K]},\alpha$}
\For{$k \in [K+1]\cup \{0\}$}
\State $\mathfrak{B}_k\gets \{   \mathcal{S} \subseteq \mathcal{V}_{f_k} :   |\mathcal{S}| \le \alpha\}$
\EndFor
\State $OPT(\{s\})\gets \{s\}$
\State $C(\{s\})\gets 0$
\For{$k \in [K+1]$}
\For{$\mathcal{S} \in \mathfrak{B}_k$}
\State $\hat{\mathcal{S}} \gets  \text{argmin}_{{\mathcal{T}}\in \mathfrak{B}_{k-1} } \{ C({\mathcal{T}}) + \frac{L_{k-1}}{\text{mincut}({\mathcal{T}};\mathcal{S})}\}$
\State $C(\mathcal{S}) \gets C(\hat{\mathcal{S}}) +\frac{L_{k-1}}{\text{mincut}(\hat{\mathcal{S}};\mathcal{S})}$
\State $OPT(\mathcal{S}) \gets (OPT(\hat{\mathcal{S}}),\mathcal{S})$
\EndFor
\EndFor
\State \textbf{return} $OPT(\{d\})$ 
\EndProcedure
\end{algorithmic}
\end{algorithm}
 
\subsection{$\alpha-$greedy Algorithm}

We notice that the complexity of the $\alpha-$optimal algorithm is not equal to a polynomial of $\alpha$. 
To solve this issue,  we introduce  $\alpha-$greedy algorithm as an approximation algorithm. 
This algorithm is exactly like Algorithm 1, except  it does not consider  subsets with more than $\alpha$ elements (see lines 23-25 of Algorithm 3).
 Details of this algorithm can be found in Algorithm 3.
 We note that the complexity of this algorithm is at most $A \times O((K+1)|\mathcal{V}|^{3})$,  like  the greedy algorithm.

\begin{algorithm}[t]
\caption{$\alpha$-greedy  Algorithm}
\begin{algorithmic}[1]
\Procedure{$\alpha$-greedy}{$\mathcal{G}, s,d,\{L_k\}_{k=0}^K \{ \mathcal{V}_{f_k}\}_{k \in [K]},\alpha$}
\For{$k \in [K+1]\cup \{0\}$}
\State $\mathcal{P}_k \gets \emptyset$
\State $\mathfrak{B}_k\gets \{   \mathcal{S} \subseteq \mathcal{V}_{f_k} :   |\mathcal{S}| = 1\}$
\EndFor
\State $D \gets \infty$
\While {$1$}
\State $OPT(\{s\})\gets \{s\}$
\State $	C(\{s\})\gets 0$
\For{$k \in [K+1]$}
\For{$\mathcal{S} \in \mathfrak{B}_k$}
\State $\hat{\mathcal{S}} \gets  \text{argmin}_{{\mathcal{T}}\in \mathfrak{B}_{k-1} } \{ C({\mathcal{T}}) + \frac{L_{k-1}}{\text{mincut}({\mathcal{T}};\mathcal{S})}\}$
\State $C(\mathcal{S}) \gets C(\hat{\mathcal{S}}) +\frac{L_{k-1}}{\text{mincut}(\hat{\mathcal{S}};\mathcal{S})}$
\State $OPT(\mathcal{S}) \gets (OPT(\hat{\mathcal{S}}),\mathcal{S})$
\EndFor
\EndFor
\If {$D \le C(\{d\})$} 
\State \textbf{break}
\EndIf
\State $ (\mathcal{P}_0,\mathcal{P}_1,\ldots, \mathcal{P}_{K+1})\gets OPT(\{d\})$ 
\State $D \gets C(\{d\})$ 
\For{$k \in [K+1]\cup \{0\}$}
\If {$|\mathcal{P}_k| < \alpha$}
\State $\mathfrak{B}_k\gets \{   \mathcal{S}\cup \mathcal{P}_k: \mathcal{S} \subseteq \mathcal{V}_{f_k} ;  |\mathcal{S}| = 1\}$
\EndIf
\EndFor
\EndWhile
\State \textbf{return} $OPT(\{d\})$ 
\EndProcedure
\end{algorithmic}
\end{algorithm}

\section{Numerical Simulations}
In this section, we evaluate the performance of the proposed algorithms by numerical simulations. 
To this end, we consider a network which is generated at random and then compare the performance of the proposed algorithms with conventional no-redundancy approach. 
We  note that although  random network construction is not necessarily the same as a given practical scenario, it is a reasonable assumption for arguing that the algorithms lead to reasonable performance in typical network scenarios.

Our setup for simulations in this section is similar to Example 2 (see Figure 3), except   we assume that the capacities are not  equal.
The reason that we choose this model for our simulations is that the performance of our algorithms  must be evaluated in the cases that varying the subsets changes the flow of the network as much as possible. 
This  makes it possible to evaluate the performance, in the worst-case scenario.
 
We assume that $L_0=L_1=\ldots=L_K$. Also we assume that the capacities of the links are chosen randomly and independently. 
The capacity of each   link is chosen uniformly from the set $(1-U,1+U)$, with probability $p$, and is set to be zero, with probability $1-p$. Here $U$ and $p$ are two parameters. For the initialization, we assume that $N=K=10$ and $\alpha = 2$. Also we set $p=1$ and $U = 0.5$. In all of the simulations, we iterate the algorithms 10 times for the independent inputs and then we consider the average performance as the output.

In Figure \ref{figN}, we examine the performance of the proposed  algorithms versus $N$. Remember that the parameter $N$  corresponds to the number of nodes that can compute a specific function. 
From  Figure \ref{figN},  we observe that the 2-greedy algorithm approaches the solution of the 2-optimal algorithm. Also, it is observed that the   greedy algorithm can improve the performance of the system when $N$ enlarges. This fact is also motivated theoretically (see Figure 4).

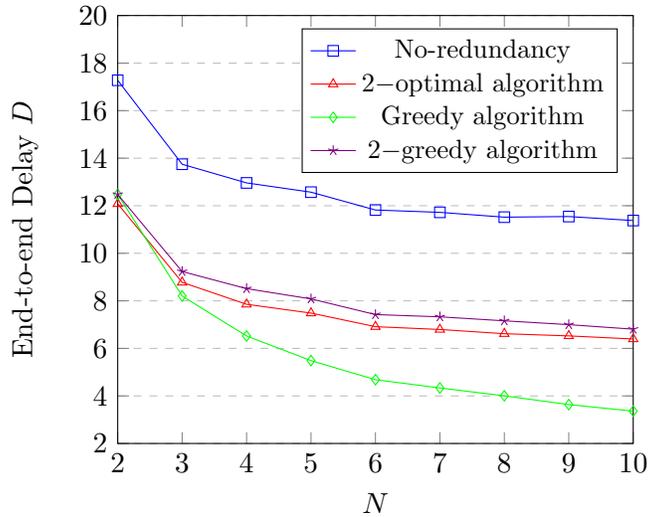
\begin{figure}[H]
\begin{center}
\begin{tikzpicture}
\begin{axis}[
    xlabel={$N$ },
    ylabel={End-to-end Delay $D$},
    xmin=2, xmax=10,
    ymin=2, ymax=20,
    xtick={2,3,4,5,6,7,8,9,10},
    ytick={2,4,6,8,10,12,14,16,18,20},
    legend pos=north east,
    ymajorgrids=true,
    grid style=dashed,
]
 \addplot[
    color=blue,
    mark=square,
    ]
    coordinates {
    (2,17.276984126984125)(3,13.745238095238093)(4,12.96031746031746)(5,12.57063492063492)(6,11.820238095238093)(7,11.721825396825395)(8,11.519444444444442)(9,11.545634920634921)(10,11.377777777777775)
    };
\addlegendentry{\small{No-redundancy}}
    
    \addplot[
    color=red,
    mark=triangle,
    ]
    coordinates {
    (2,12.080579589828817)(3,8.778197859931606)(4,7.854546549941287)(5,7.48428756915599)(6,6.912750626566417)(7,6.794770013268464)(8,6.616993842071241)(9,6.525538724261634)(10,6.39488256841198)
    };
\addlegendentry{\small{$2-$optimal algorithm}}

\addplot[
    color=green,
    mark=diamond,
    ]
    coordinates {
    (2,12.476146114807104)(3,8.203433038721528)(4,6.520314246334911)(5,5.485011378779412)(6,4.682878023439423)(7,4.33490869458815)(8,4.005390115716918)(9,3.635625750025577)(10,3.361662206200597)
    };
\addlegendentry{\small{Greedy algorithm}}

\addplot[
    color=violet,
    mark=star,
    ]
    coordinates {
    (2,12.476146114807104)(3,9.232312007737704)(4,8.51325173364647)(5,8.081687073676239)(6,7.421703910984097)(7,7.327061947340584)(8,7.159199821506322)(9,6.997939229672975)(10,6.807594495287994)
    };
\addlegendentry{\small{$2-$greedy algorithm}}
 \end{axis}
\end{tikzpicture}
\end{center}
\caption{
End-to-end delay versus $N$.}
\label{figN}
\end{figure}

In Figure \ref{figK}, we investigate the performance of the proposed algorithms versus $K$ which  is the number of functions that must be  computed in the chain.
From Figure 6, the result of the 2-greedy algorithm approaches  to the 2-optimal algorithm. Also, the greedy algorithm outperform one-element sets or two-elements.

\begin{figure}[H]
\begin{center}
\begin{tikzpicture}
\begin{axis}[
    xlabel={$K$ },
    ylabel={End-to-end Delay $D$},
    xmin=2, xmax=10,
    ymin=1, ymax=14,
    xtick={2,3,4,5,6,7,8,9,10},
    ytick={2,4,6,8,10,12,14},
    legend pos=north west,
    ymajorgrids=true,
    grid style=dashed,
]
 \addplot[
    color=blue,
    mark=square,
    ]
    coordinates {
    (2,3.328571428571428)(3,4.225)(4,5.388492063492063)(5,6.265079365079365)(6,7.299999999999999)(7,8.322222222222221)(8,9.30674603174603)(9,10.288888888888888)(10,11.37063492063492)
    };
\addlegendentry{\small{No-redundancy}}
    
    \addplot[
    color=red,
    mark=triangle,
    ]
    coordinates {
    (2,2.298140335756435)(3,2.731744188903632)(4,3.253818369453044)(5,3.780982480711582)(6,4.257378740970071)(7,4.816787065703474)(8,5.390879772961815)(9,5.813121037888839)(10,6.399833362372061)
    };
\addlegendentry{\small{$2-$optimal algorithm}}

\addplot[
    color=green,
    mark=diamond,
    ]
    coordinates {
    (2,1.872268295251292)(3,2.067724237774228)(4,2.152179321965094)(5,2.378707995565981)(6,2.439293367491833)(7,2.772939523965115)(8,2.963873100614916)(9,2.991066219624847)(10,3.232640545297129)
    };
\addlegendentry{\small{Greedy algorithm}}

\addplot[
    color=violet,
    mark=star,
    ]
    coordinates {
    (2,2.466673702974012)(3,2.938941244514)(4,3.450738979802447)(5,3.973025228984981)(6,4.520599056086672)(7,5.131762617327632)(8,5.654681575734206)(9,6.089511297096435)(10,6.754581752956365)
    };
\addlegendentry{\small{$2-$greedy algorithm}}
 \end{axis}
\end{tikzpicture}
\end{center}
\caption{
End-to-end delay versus $K$.}
\label{figK}
\end{figure}
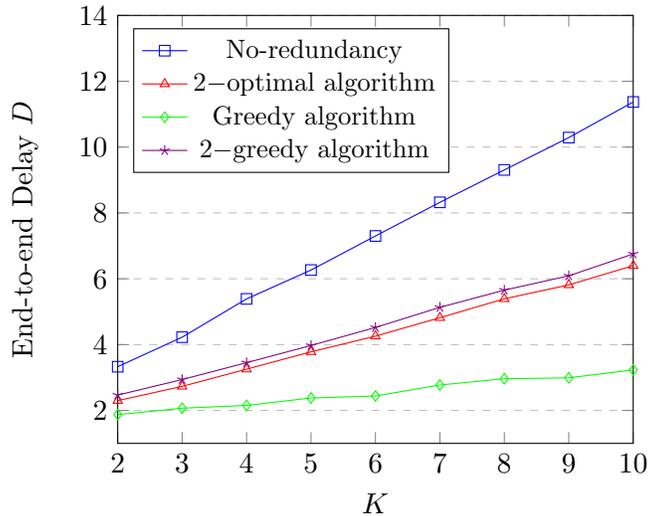

 Figure \ref{figa} demonstrates the end-to-end delay versus the parameter $\alpha$. It is observed that the greedy algorithm is near optimal, where for large $\alpha$, the performance of the greedy algorithm is close to the optimal case.

\begin{figure}[H]
\begin{center}
\begin{tikzpicture}
\begin{axis}[
    xlabel={$\alpha$ },
    ylabel={End-to-end Delay $D$},
    xmin=2, xmax=10,
    ymin=1, ymax=12,
    xtick={2,3,4,5,6,7,8,9,10},
    ytick={2,4,6,8,10,12},
    legend pos=north west,
    ymajorgrids=true,
    grid style=dashed,
]
 \addplot[
    color=blue,
    mark=square,
    ]
    coordinates {
    (2,7.468864468864469)(3,7.47021312021312)(4,7.480769230769231)(5,7.51080586080586)(6,7.497136197136197)(7,7.516117216117218)(8,7.486996336996339)(9,7.536996336996339)(10,7.462637362637364)
    };
\addlegendentry{\small{No-redundancy}}

\addplot[
    color=green,
    mark=diamond,
    ]
    coordinates {
    (2,1.91233177434261)(3,1.839932492032854)(4,1.905964265379947)(5,1.861576657675898)(6,1.883772594225156)(7,1.843952490465797)(8,1.878315203118342)(9,1.880817666635403)(10,1.894997678399945	)
    };
\addlegendentry{\small{Greedy algorithm}}

\addplot[
    color=violet,
    mark=star,
    ]
    coordinates {
    (2,4.285062972454278)(3,3.213580793257941)(4,2.751046052351586)(5,2.40179208932679)(6,2.222738478106111)(7,2.065255381227234)(8,1.990829096205961)(9,1.928057326122624)(10,1.894997678399945)
    };
\addlegendentry{\small{$\alpha-$greedy algorithm}}
 \end{axis}
\end{tikzpicture}
\end{center}
\caption{
End-to-end delay versus $\alpha$.}
\label{figa}
\end{figure}
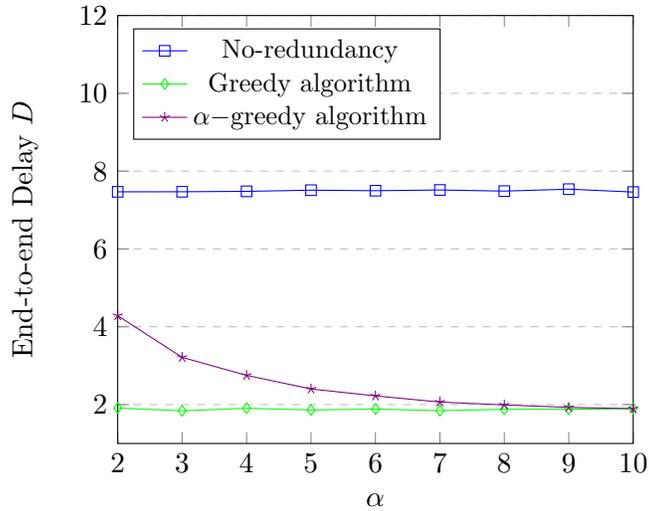

Figure \ref{figp} demonstrates the performance of the proposed algorithms versus $p$. Note that $p$   corresponds to the connectivity of the network. The simulations show that the proposed algorithms have a good performance, even when the network is sparse (small values for $p$).

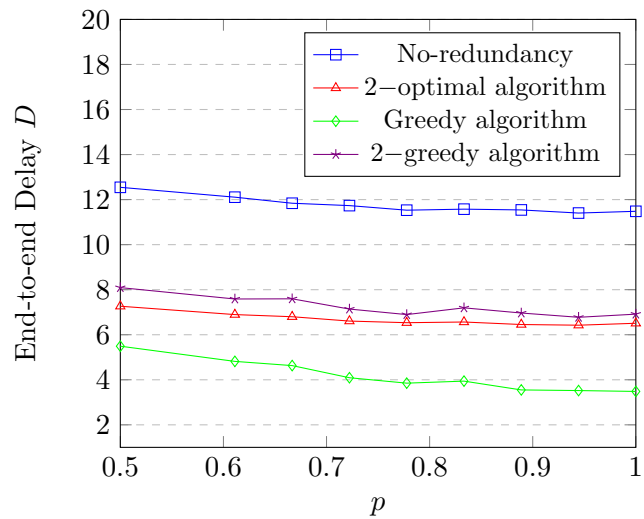
\begin{figure}[H]
\begin{center}
\begin{tikzpicture}
\begin{axis}[
    xlabel={$p$ },
    ylabel={End-to-end Delay $D$},
    xmin=0.5, xmax=1,
    ymin=1, ymax=20,
    xtick={0.5,0.6,0.7,0.8,0.9,1},
    ytick={2,4,6,8,10,12,14,16,18,20},
    legend pos=north east,
    ymajorgrids=true,
    grid style=dashed,
]
 \addplot[
    color=blue,
    mark=square,
    ]
    coordinates {
    (0.5,12.545238095238094)
    (0.611111111111111,12.107539682539683)
    (0.666666666666667,11.838492063492062)
    (0.722222222222222,11.733333333333333)
    (0.777777777777778,11.530158730158732)(0.833333333333333,11.57579365079365)(0.888888888888889,11.541269841269841)(0.944444444444444,11.402777777777779)(1,11.480555555555554)
    };
\addlegendentry{\small{No-redundancy}}
    
    \addplot[
    color=red,
    mark=triangle,
    ]
    coordinates {
    (0.5,7.265870322487968)(0.611111111111111,6.894866975130133)(0.666666666666667,6.799686480530133)(0.722222222222222,6.606935230232443)(0.777777777777778,6.53477849979398)(0.833333333333333,6.565316855487133)(0.888888888888889,6.453922560927206)(0.944444444444444,6.42373273870952)(1,6.509583158499568)
    };
\addlegendentry{\small{$2-$optimal algorithm}}

\addplot[
    color=green,
    mark=diamond,
    ]
    coordinates {
    (0.5,5.491830295811345)(0.611111111111111,4.817223447284746)(0.666666666666667,4.630304758863023)(0.722222222222222,4.091143187408703)(0.777777777777778,3.850823296222963)(0.833333333333333,3.940923682293797)(0.888888888888889,3.550051540530008)(0.944444444444444,3.521706945251955)(1,3.482850913569882)
    };
\addlegendentry{\small{Greedy algorithm}}

\addplot[
    color=violet,
    mark=star,
    ]
    coordinates {
    (0.5,8.093291151572886)(0.611111111111111,7.587255842712497)(0.666666666666667,7.595105264532509)(0.722222222222222,7.140131720704476)(0.777777777777778,6.895179502300245)(0.833333333333333,7.19416655326098)(0.888888888888889,6.964442355889723)(0.944444444444444,6.775026555831509)(1,6.914310049255869)
    };
\addlegendentry{\small{$2-$greedy algorithm}}
 \end{axis}
\end{tikzpicture}
\end{center}
\caption{
End-to-end delay versus $p$.}
\label{figp}
\end{figure}

At the end, we examine the algorithms for the case that the capacities have large variance, i.e., $U$ is set to be large. It is observed that  in this case the performance of the algorithms does not decay and they have  gains near to the theoretical analysis.
 This shows that the performance of the algorithms is not sensitive to the homogeneity of the capacities.

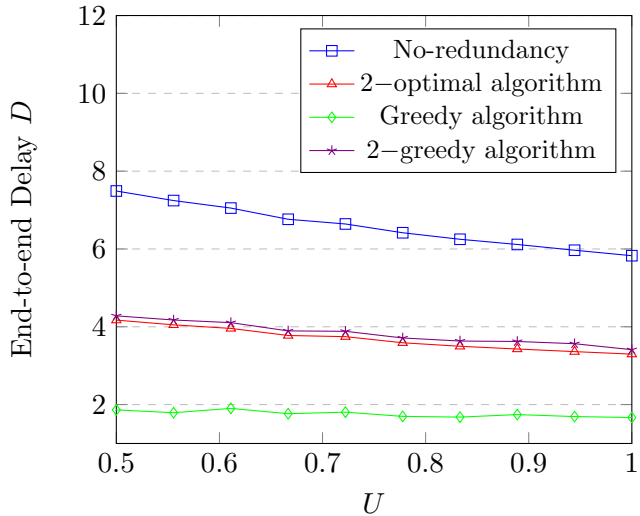
\begin{figure}[H]
\begin{center}
\begin{tikzpicture}
\begin{axis}[
    xlabel={$U$ },
    ylabel={End-to-end Delay $D$},
    xmin=0.5, xmax=1,
    ymin=1, ymax=12,
    xtick={0.5,0.6,0.7,0.8,0.9,1},
    ytick={2,4,6,8,10,12},
    legend pos=north east,
    ymajorgrids=true,
    grid style=dashed,
]
 \addplot[
    color=blue,
    mark=square,
    ]
    coordinates {
    
    (0.5,7.48981018981019) 
    (0.555555555555556,7.243991898484093)
    (0.611111111111111,7.050797476793683)
    (0.666666666666667,6.763176881597936)
    (0.722222222222222,6.642133627557198)
    (0.777777777777778,6.416502463054185)
    (0.833333333333333,6.248868380575697)
    (0.888888888888889,6.115433755816278)
    (0.944444444444444,5.967323014662635)
    (1,5.826775008892151)
    };
\addlegendentry{\small{No-redundancy}}
    
    \addplot[
    color=red,
    mark=triangle,
    ]
    coordinates {
   (0.5,4.170960489063938) 
    (0.555555555555556,4.049211932638372)
    (0.611111111111111,3.957295233748218)
    (0.666666666666667,3.775849318666918)
    (0.722222222222222,3.745700413205633)
    (0.777777777777778,3.588760775074507)
    (0.833333333333333,3.4982796151259)
    (0.888888888888889,3.428353527978326)
    (0.944444444444444,3.361077919199889)
    (1,3.296994506701346)
    };
\addlegendentry{\small{$2-$optimal algorithm}}

\addplot[
    color=green,
    mark=diamond,
    ]
    coordinates {
    (0.5,1.86300281582352) 
    (0.555555555555556,1.791109098043571)
    (0.611111111111111,1.90140701192694)
    (0.666666666666667,1.765589036375048)
    (0.722222222222222,1.80471396392211)
    (0.777777777777778,1.696591607530464)
    (0.833333333333333,1.678013653239156)
    (0.888888888888889,1.742919117364247)
    (0.944444444444444,1.691998349702906)
    (1,1.666367013648918)
    };
\addlegendentry{\small{Greedy algorithm}}

\addplot[
    color=violet,
    mark=star,
    ]
    coordinates {
	(0.5,4.280218134356066) 
    (0.555555555555556,4.173937779186114)
    (0.611111111111111,4.10612821715646)
    (0.666666666666667,3.893701195758189)
    (0.722222222222222,3.883119949163993)
    (0.777777777777778,3.712207897588266)
    (0.833333333333333,3.633325875151576)
    (0.888888888888889,3.621961401202838)
    (0.944444444444444,3.564605014378294)
    (1,3.410739775998693)
    };
\addlegendentry{\small{$2-$greedy algorithm}}
 \end{axis}
\end{tikzpicture}
\end{center}
\caption{
End-to-end delay versus $U$.}
\label{figU}
\end{figure}

For the end of this section, we note that our proposed greedy algorithms are not optimal in general, but our simulations reveal that they have acceptable performance for most scenarios.

\section{Conclusion and Discussion}
In this paper,  we investigated the fundamental limits of end-to-end delay minimization in NFV-based networks. 
It was observed that traditional no-redundancy approaches are not enough for achieving the optimum delay since there are some  examples that show the inefficiency of them. 
To exploit the capacity of such networks, in this paper, it has been assumed that to respond a request,  the functions can be computed redundantly in parallel, in addition to use the network coding in the system. 
It was showed that redundantly function computation and using network coding for delivery, not only improves the reliability measures, but also decreases the transmission delay.
In some cases, this gain would also scale  by the size of the network,
in comparison with the traditional no-redundancy approaches.
The sufficiency of linear codes to exploit the capacity of the system was also shown.
 Then, it was observed that the problem of finding the subsets of the nodes that must compute the functions at each round  in general is related to a complex integer programming problem. 
 To this end,  an approximation algorithm  was  proposed.
 The optimal algorithm with polynomial  complexity for the case that there is an upper bound on the number of nodes computing functions in each round was also provided.
 The performances of the algorithms were evaluated through numerical simulations in several cases, where this was showed that they reach order-wise improvements in comparison with no-redundancy approaches.

For  future work, a number of problems might be of interest.  
First, we note that in this paper, we proposed a method to decrease the end-to-end delay through a theoretical modeling, and  argued that our algorithms lead to  good performance for typical  network topologies.
In this direction, an important  question is to evaluate performance under more realistic network models.
As an example, an interesting direction is to consider central cloud-edge cloud scenarios where the computing capacities of the nodes and link capacities may vary significantly, and in some scenarios, lead to less connectivity in the network. 
Therefore, it is fair to say that   understanding the benefit of redundant computing  in specific network topologies  remains  an open problem.

On another direction, we note that in this paper, we have considered the class of networks which are described by directed graphs. 
However, undirected network models have also been of interest.   
A natural question is that what will be happen if we consider such network models?
We note that if the network consists of two-way links with fixed capacities, then the results of this paper are still valid.
However, if we have constraints on the sum capacities of undirected links, then the problem is challenging and can be addressed as future work. 
Another future direction is to consider more than one request which must be served simultaneously.

We should also  note that from a practical  point of view, our  results suggest that one should consider the trade-off between delay and computation/communication cost in real-world designs. 
We suggest that for each network topology and network functions class, the achievable gain in decreasing the delay, in terms of computation/communication cost should be evaluated. 
Also there is a number of challenges for implementing the redundant computation. 
For example, we note that in this paper, the servers are assumed to be synchronized.
However, they may have different processing times and 
this leads to its own implementation challenges.

Another interesting question that can be investigated in future is how to address the placement of  VNFs, or  design of the network infrastructure, in order to potentially achieve high gains through the proposed model of  redundant computation?
The answer to such question would naturally 	have significant effect on choice of system parameters for a given practical application and the resulting quality of service and experienced at    the user side.

\section*{Acknowledgment}
This work was supported in part by a grant from Institute for Research in Fundamental Sciences (IPM).

\appendix
\section{Proof of Theorem 1}
For the proof of Theorem 1, we use the single-source multicast   theorem \cite{code7}. 
First we add  an extra node, denoted by $EN$, to the infrastructure graph $\mathcal{G}$ and form a  new directed graph $\tilde{\mathcal{G}}$.
In $\tilde{\mathcal{G}}$,  assume that $EN$ is connected to each node of    $\mathcal{S}_{k-1}$ via directed links with infinite capacity.

Consider a  single-multicast problem from $EN$ to all of the nodes in $\mathcal{S}_k$. 
We claim that the resulting auxiliary multicast problem is equivalent to the original problem.
To show this fact, first we note that each solution of the auxiliary multicast problem, is also a solution of the original  problem. 
This is due  to the fact that for each solution of the auxiliary problem, the source node $EN$ first sends the data to the nodes in $\mathcal{S}_{k-1}$, because they are  neighbors of $EN$ in $\tilde{\mathcal{G}}$. 
This procedure does not experience  any delay, because the corresponding links have infinite capacities. 
Subsequently, the nodes in $\mathcal{S}_{k-1}$ deliver the data to the destination nodes $\mathcal{S}_{k}$.
For converse, consider a solution for the original problem. 
We construct a solution for the auxiliary problem with the same delay as the original problem.
Assume that first the node $EN$ transmits  all the data to each node in $\mathcal{S}_{k-1}$. Again, note that this procedure does  not have any delay, due to the infinite capacity of links. Consequently, we apply the solution of the original problem to the auxiliary problem, such that nodes in
$\mathcal{S}_{k-1}$ deliver their common data to each node in $\mathcal{S}_k$.
This completes the proof of the equivalency of two problems.

Now we apply the  single-source multicast   theorem \cite{code7} to the auxiliary problem to compute the minimum delay. 
In the network coding terminology, this theorem shows that the capacity of such system is equal to the capacity of min-cut  in $\tilde{\mathcal{G}}$.
Hence, we have
\begin{align}
D_k& \ge  
 \frac{L_{k-1}}{\min\limits_{v \in \mathcal{S}_k} \text{mincut}(\{EN\};\{v\})}\\
&= \frac{L_{k-1}}{\min\limits_{v \in \mathcal{S}_k}\min\limits_{
\substack{
            \mathcal{S}\subseteq \mathcal{V}\cup\{EN\} \\
             EN \in \mathcal{S} \\
             v\notin \mathcal{S}
             }
} \cut(\mathcal{S})},
\end{align}
where the notation ``mincut'' is defined in (12). 
Now we notice that if $\mathcal{S}_{k-1}\not \subseteq \mathcal{S}$, then $\cut(\mathcal{S})$ is equal to infinity since the capacity of directed links between $EN$ and each node in $\mathcal{S}_{k-1}$ is infinity. Hence, we have
\begin{align}
D_k &\ge \frac{L_{k-1}}{\min\limits_{v \in \mathcal{S}_k}\min\limits_{
\substack{
            \mathcal{S}_{k-1} \subseteq \mathcal{S}\subseteq \mathcal{V}\cup\{EN\} \\
             EN \in \mathcal{S} \\
             v\notin \mathcal{S}
             }
} \cut(\mathcal{S})} \\
& = \frac{L_{k-1}}{\min\limits_{
\substack{
	\mathcal{S}_{k-1} \subseteq \mathcal{S} \subseteq \mathcal{V}\\
	  \mathcal{S}_{k} \not \subseteq \mathcal{S}
	  }
	  } \cut(\mathcal{S})}.
\end{align}
This completes the proof. We emphasize that the capacity of the single-source multicast   problem can be achieved using linear codes \cite{code10}. This proves the achievability of the lower bound, by using the optimal linear network codes in the corresponding transmission round. Note that in the single-source multicast   theorem, for achieving the capacity, it is essential to use the network coding and routing is not generally optimum. 
In addition, the message size (or equivalently $L_{k-1}$) needs to be large enough, in order to ensure existence of capacity achieving network codes.

\section{Proof of Theorem 2}
Based on Theorem 1, we have
\begin{align}
D  = \sum_{k=1}^{K+1} D_k\ge \sum_{k=1}^{K+1} \frac{L_{k-1}}{\min\limits_{
\substack{
	\mathcal{S}_{k-1} \subseteq \mathcal{S} \subseteq \mathcal{V}\\
	  \mathcal{S}_{k} \not \subseteq \mathcal{S}
	  }
}  \cut(\mathcal{S})}.
\end{align}
Also, the above lower bound is  tight, by using the achievable scheme of Theorem 1 in all rounds of data transmission. 
Hence, the above achievable  lower bound on  end-to-end delay just depends on the subsets $\{\mathcal{S}_k\}_{k \in [K]}$. This means that the optimum delay can be achieved by minimizing the above equation as a function of the subsets. This completes the proof.

\section{Proof of Lemma 1}
First notice that for $\mathcal{S} = \mathcal{S}_{k-1}$, we have
$$\cut(\mathcal{S}) = \epsilon \times |\mathcal{S}_{k-1}|\times (|\mathcal{V}|-|\mathcal{S}_{k-1}|).$$
In order to prove the lemma, it suffices to show that for any $\mathcal{S}\subseteq \mathcal{V}$ such that     $\mathcal{S}_{k-1} \subseteq \mathcal{S}$ and $ \mathcal{S}_{k} \not \subseteq \mathcal{S} $, we have $\cut(\mathcal{S}) \ge \epsilon \times |\mathcal{S}_{k-1}|\times (|\mathcal{V}|-|\mathcal{S}_{k-1}|).$

Note that if $ |\mathcal{S}_{k-1}| \le|\mathcal{S}| \le |\mathcal{V}| - |\mathcal{S}_{k-1}|$, then we have
\begin{align}
\cut(\mathcal{S}) &\ge \epsilon \times |\{e = (u,v) \in \mathcal{E}: u \in \mathcal{S};v \not \in \mathcal{S}\}|\\
&= \epsilon \times |\mathcal{S}| \times (|\mathcal{V}| - |\mathcal{S}|)\\
&\ge \epsilon \times |\mathcal{S}_{k-1}| \times (|\mathcal{V}| - |\mathcal{S}_{k-1}|).
\end{align}
and in this case the proof is completed.

Now consider the case that $ |\mathcal{S}| < |\mathcal{S}_{k-1}| $ or $ |\mathcal{S}| >  |\mathcal{V}| - |\mathcal{S}_{k-1}|$. The case  
$ |\mathcal{S}| < |\mathcal{S}_{k-1}| $
is impossible, due to the fact that $\mathcal{S}_{k-1} \subseteq \mathcal{S}$. Hence, we assume that $ |\mathcal{S}| >  |\mathcal{V}| - |\mathcal{S}_{k-1}| \ge KN+N+1-N = KN+1$. This means that $ \mathcal{S}$ contains at least one relay node which is distinct of $s$, such as $\tilde{u} \in \mathcal{S}$. 
Since $ \mathcal{S}_{k} \not \subseteq \mathcal{S} $, there is a node $\tilde{v} \in \mathcal{S}_{k}$ such that $\tilde{v} \not \in \mathcal{S}$.
Let us define $\tilde{e}= (\tilde{u},\tilde{v})$.  Therefore, we write
\begin{align}
\cut(\mathcal{S}) \ge w(\tilde{e}) = 1\overset{(a)}{\ge} \epsilon \times |\mathcal{S}_{k-1}| \times (|\mathcal{V}| - |\mathcal{S}_{k-1}|),
\end{align}
where (a) follows from the fact that $\epsilon$ is set to be small enough. 
This completes the proof.

\end{document}